\documentclass[11pt]{article}
\usepackage[margin=0.9in]{geometry}
\usepackage{amsthm,amsmath,amssymb,amscd,enumerate, graphicx,subcaption,pstricks,bm,booktabs}
\usepackage{setspace,array}
\usepackage{thmtools}
\usepackage{thm-restate}
\usepackage{color}
\usepackage{natbib}
\usepackage[mathscr]{euscript}
\usepackage[shortlabels]{enumitem}
\usepackage[colorlinks=true, urlcolor=blue, linkcolor=blue, citecolor=blue]{hyperref}
\usepackage{algpseudocode}
\usepackage{algorithm}
\newcolumntype{x}[1]{>{\centering\arraybackslash\hspace{0pt}}p{#1}}
\algrenewcommand\algorithmicindent{0.4em}

\DeclareMathAlphabet{\mathbbs}{U}{bbm}{m}{sl}
\newcolumntype{C}[1]{>{\centering\let\newline\\\arraybackslash\hspace{0pt}}m{#1}}
\declaretheorem[name=Theorem]{theorem}
\declaretheorem[name=Lemma]{lemma}

\allowdisplaybreaks

\DeclareMathOperator*{\argmin}{argmin}
\DeclareMathOperator*{\argmax}{argmax}
\DeclareMathOperator*{\logit}{logit}
\DeclareMathOperator*{\expit}{expit}

\newcommand{\1}[1]{\mathbf{1}\{#1\}}

\usepackage{natbib}
\usepackage{xcolor}
\usepackage{colortbl}
\usepackage[normalem]{ulem}
\usepackage[T1]{fontenc}
\usepackage[utf8]{inputenc}
\usepackage{authblk}

\definecolor{Gray}{gray}{0.95}

\newcommand{\convset}{B}

\newcommand{\pr}{\text{pr}}

\newcommand{\nrvs}{U}

\newcommand{\distv}{Q}

\newcommand{\intvar}{ t}

\newcommand{\garg}{ x}
\newcommand{\gengam}[1][]{\Gamma_{#1}(\garg, \Sigma_0, \varphi)}
\newcommand{\pregam}{\zeta}

\title{A general adaptive framework for\\multivariate point null testing}
\author[1]{Adam Elder}
\author[1,2,3]{Marco Carone}
\author[3,1]{Peter Gilbert}
\author[1,2,3]{Alex Luedtke}
\affil[1]{Department of Biostatistics, University of Washington}
\affil[2]{Department of Statistics, University of Washington}
\affil[3]{Vaccine and Infectious Disease Division, Fred Hutchinson Cancer Research Center}

\begin{document}

\maketitle

\doublespacing

\begin{abstract}
	As a common step in refining their scientific inquiry, investigators are often interested in performing some screening of a collection of given statistical hypotheses. For example, they may wish to determine whether any one of several patient characteristics are associated with a health outcome of interest. Existing generic methods for testing a multivariate hypothesis --- such as multiplicity corrections applied to individual hypothesis tests --- can easily be applied across a variety of problems but can suffer from low power in some settings. Tailor-made procedures can attain higher power by building around problem-specific information but typically cannot be easily adapted to novel settings. In this work, we propose a general framework for testing a multivariate point null hypothesis in which the test statistic is adaptively selected to provide increased power. We present theoretical large-sample guarantees for our test under both fixed and local alternatives. In simulation studies, we show that tests created using our framework can perform as well as tailor-made methods when the latter are available, and we illustrate how our procedure can be used to create tests in two settings in which tailor-made methods are not currently available. 
\end{abstract}

\section{Introduction}

Addressing a scientific question often involves performing simultaneous inference on components of a vector-valued statistical parameter and, in particular, assessing whether this parameter deviates from a specific null value of scientific interest. Indeed, testing of a multivariate point null hypothesis arises commonly in applications. For example, it may be of interest to determine whether any of several variables is related to a particular health outcome, as often occurs in genetics \citep{gao_multiple_2008}, neurology \citep{flandin_analysis_2019}, and vaccine development \citep{borthwick_vaccine-elicited_2014}, among other fields. General-purpose strategies (e.g., construction of Wald-type test statistics) exist for performing a hypothesis test of a univariate point null with a specified (asymptotic) type I error; in many cases, such strategies can be shown to yield optimal tests. The corresponding problem for a multivariate point null poses a much greater challenge.

A valid test of a multivariate null hypothesis can be constructed on the basis of multiple tests of univariate null hypotheses in a manner that controls the family-wise type I error rate. For decades, the Bonferroni correction has been used to derive multiple hypothesis testing procedures.  Early examples of its use appear in  \cite{dunn_estimation_1959,dunn_multiple_1961}. Refinements of the Bonferroni correction have been proposed by various authors, including, for example,  \cite{holm_simple_1979}, \cite{simes_improved_1986}, \cite{hommel_stagewise_1988}, \cite{hochberg_sharper_1988} and \cite{s._holland_improved_1988}.  Bonferroni-type correction procedures are broadly applicable and  easily implemented. However, because they do not leverage knowledge of the dependence between the test statistics involved, they may yield low power in some circumstances. Some authors, including \citet{lehmann_testing_2005} and  \citet{dudoit_multiple_2008}, have proposed alternative strategies to mitigate this problem by accounting for the joint behavior of the test statistics. These procedures, in particular, allow users to specify the desired trade-off between type I and II errors by controlling, for example, the false discovery rate or family-wise error rate of the test. Nevertheless, despite these improvements, the use of multiple testing techniques to assess a single multivariate hypothesis, while convenient, comes at a price. The ability to determine which null hypothesis (if any) to reject, while potentially valuable, could come at the cost of lower power for detecting deviations from the multivariate point null. Indeed, for any multiple testing procedure that achieves family-wise type I error control, there exists a calibrated test of the multivariate null with at least as much power. In fact, a more powerful test of the multivariate null would be expected to exist since such a test does not need to account for rejections of a univariate null that holds when others do not.

Approaches for multivariate testing have been proposed and typically account for the correlation between individual test statistics.  Such methods can be categorized based on how an aggregate test statistic is constructed. In some procedures (e.g.,  \citealp{donoho_higher_2004}), a summary test statistic is built using estimators of underlying univariate parameters, whereas in others (e.g., \citealp{liu_cauchy_2020})  $p$-values from multiple univariate tests are directly combined. Unfortunately, these procedures are usually tailored to a specific parameter and statistical model (e.g., \citealp{donoho_higher_2004}) or make assumptions about the data-generating mechanism that can fail in practice (e.g., sparsity conditions, parametric modeling assumptions). Additionally, some procedures do not allow the use of flexible learning strategies in the construction of the involved test statistics \citep{breiman_statistical_2001}.  While the use of flexible learners is often critical to obtaining asymptotic guarantees in nonparametric and semiparametric models, it can also cause poor finite-sample performance of testing procedures, especially when the adaptive nature of the test statistic is not taken into account (see, e.g., \citealp{leeb_model_2005, leeb_can_2006}). While more recent proposals address many of these potential issues  (e.g., \citealp{pan_powerful_2014, mckeague_adaptive_2015, xu_adaptive_2016}), they provide techniques for use in specific applications rather than general-purpose templates for use in a variety of problems. Thus, while procedures for multivariate testing with good performance characteristics have been devised for certain settings, in many cases, there is little guidance for investigators beyond crude approaches such as the Bonferroni correction. In this paper, we propose and study a general-purpose procedure for constructing a test of a multivariate point null hypothesis that can be used for a broad range of statistical parameters and models. Our procedure benefits from an explicit accounting of the joint behavior of the test statistic, and incorporates data-driven selection of the involved tuning parameters to optimize test performance for the application at hand. As such, it can be expected to provide improved performance compared to existing strategies in many contexts.
	
This paper is organized as follows. In Section \ref{sec:Working Examples}, we introduce the testing problem considered and provide working examples with which we will illustrate the implementation and performance of our proposed procedure. We formally describe our procedure in Section \ref{sec:prop_test_proc}, and provide a theoretical study of its properties in Section \ref{sec:theorems}. In Section \ref{sec:sim_stdy}, we illustrate through simulation studies that the proposed framework yields novel tests with comparable power to tailor-made procedures in settings in which specialized methods already exist, and has good operating characteristics in settings in which problem-specific methods do not currently exist. In Section \ref{sec:data_app}, we use our procedure to test for the existence of a correlate of risk of HIV infection using data from the HVTN 505 HIV vaccine trial. In Section \ref{sec:discuss}, we provide concluding remarks. Technical proofs as well as additional simulation results and details on our data analysis are provided in the Supplement.
	
\section{Problem setup}

\label{sec:Working Examples}

Suppose that we have at our disposal observations $X_1,X_2,\ldots,X_n$ drawn independently from a common unknown distribution $P_0\in\mathcal{M}$, where the statistical model $\mathcal{M}$ encodes known restrictions on $P_0$. In the developments below, we are primarily interested in cases in which $\mathcal{M}$ is a nonparametric or semiparametric model, although this is not a requirement for the developments presented. We denote by $\mathcal{X}$ the union of the support of $P$ for each $P\in\mathcal{M}$. Suppose that $\Psi_1,\Psi_2,\ldots,\Psi_d$ form a collection of real-valued statistical parameters defined on $\mathcal{M}$. For each $j\in\{1,2,\ldots,d\}$, we define $\psi_{j0}:=\Psi_j(P_0)\in\mathbb{R}$ to be the evaluation of $\Psi_j$ on $P_0$, and write $\psi_0:=(\psi_{10},\psi_{20},\ldots,\psi_{d0})$. In this article, for a given (known) vector $\psi_*:=(\psi_{1*},\psi_{2*},\ldots,\psi_{d*})\in\mathbb{R}^d$, we consider testing \begin{align}\label{hyp}H_0:\psi_0=\psi_{*}\text{\ \ versus\ \ }H_1:\psi_0\neq \psi_{*}\ .\end{align} Without loss of generality, we consider the case $\psi_{*}=(0,0,\ldots,0)$ since otherwise we may instead take $\Psi_j$ to be its null-centered counterpart $P\mapsto \Psi_j(P)-\psi_{j*}$.

The setup we consider is sufficiently broad to include a large variety of examples. For concreteness, we present here three particular examples that we will use throughout as an illustration of our general results.
	
\noindent\textbf{Example 1: correlation.} In our first and simplest example, we consider the data unit $X=(W,Y)$, where $W:=(W_1,W_2,\ldots,W_d)$ represents a vector of real-valued covariates and $Y$ is some outcome of interest, and the parameter of interest $\Psi_j(P):=\text{corr}_P(W_j,Y)$ is the marginal correlation between $W_j$ and $Y$ under $P$. We are interested in testing the multivariate null hypothesis that none of the components of $W$ are marginally correlated with $Y$ in a nonparametric model. For this problem, there exist several competing approaches in the literature, and we will compare a test derived using our proposal to several of these existing approaches.

\noindent\textbf{Example 2: coefficients of a working log-linear regression model under missingness.} In our second example, we instead consider the data unit $X=(W,U,\Delta)$, where $W:=(W_1,W_2,\ldots,W_d)$ again represents a vector of real-valued covariates, $\Delta$ is an indicator that the binary outcome  $Y$ is observed, and $U:=\Delta Y$ equals $Y$ if $\Delta=1$ and is set to zero otherwise. In other words, this data unit is similar to that defined in Example 1 but with the outcome value possibly missing. We focus here on coefficients indexing the least-squares projection of the true conditional success probability onto the log-linear regression model $\log \pr(Y=1\,|\,W_j=w_j)=\alpha_0+\alpha_jw_j$. Assuming missingness at random, that is, that $Y$ and $\Delta$ are independent conditionally upon $W$, the parameter
\begin{align}
	\Psi_j(P):=\frac{\text{cov}_P\left[W_j,\,\log E_P\left\{P(U = 1\,|\,\Delta = 1, W)\,|\,  W_j\right\}\right]}{\text{var}_P(W_j)} \label{eqn:ident_exp_2_param}
\end{align}
identifies the coefficient associated to $W_j$ in the projection onto the log-linear working model, and simplifies to $\alpha_j$ when this working model holds true. This parameter represents a measure of association between positive outcome $Y$ and covariate $W_j$ for use when $Y$ is possibly missing at random given $W$.  We are interested in testing, within a nonparametric model, the multivariate null hypothesis that all coefficients of this working log-linear model equal zero.

\noindent\textbf{Example 3: coefficients of a working effect modification model for randomized trials.} In our third example, we consider the data unit $X=(W,A,Y)$, where $W:=(W_1,W_2,\ldots,W_d)$ once more represents a vector of real-valued covariates, $A\in\{0,1\}$ is a binary treatment variable, and $Y$ is a binary outcome of interest, and focus on the interaction coefficient of the least-squares projection of the true conditional success probability onto the logistic model  $\logit\pr\left(Y=1\,|\, W_j=w, A_j=a\right)=\alpha_{0j} + \alpha_{1j} a + \alpha_{2j} w + \delta_j w a$. This coefficient provides a measure of the degree to which $W_j$ modifies the effect of $A$ on $Y$ in a randomized trial. The parameter of interest can be expressed as
\begin{align*}
\Psi_j(P):= \argmin_{\gamma}\min_{ \alpha}E_P\left[\logit \,P(Y=1\,|\,A,W_j)\,-\alpha_{0}-\alpha_{1} A-\alpha_{2} W_j-\gamma W_jA\right]^2 \ ,
\end{align*}which identifies the interaction coefficient in this working model, and simplifies to $\delta_j$ when the working logistic model above holds.  Once more, we are interested in testing, within a nonparametric model, the multivariate null hypothesis that each $\Psi_j(P)$ is equal to zero.

\section{Proposed testing procedure}
\label{sec:prop_test_proc}

\subsection{Non-adaptive test}

While the test we ultimately propose is adaptive, it can be viewed as a refinement of non-adaptive counterparts, which we begin by describing. We define $\mathcal{M}_0:=\{P\in\mathcal{M}:\Psi_j(P)=0\text{ for each }j=1,2,\ldots,d\}$ to be the collection of all distributions in $\mathcal{M}$ under which the null hypothesis \eqref{hyp} is true. Suppose that an estimator $\psi_n:=(\psi_{1n},\psi_{2n},\ldots,\psi_{dn})$ of $\psi_0$ is available, and that for each $P\in\mathcal{M}$, $n^{1/2}(\psi_n-\psi_0)$ tends in distribution to a random vector $U_0$ following the $d$-dimensional normal distribution $Q_0$ with mean zero and positive definite covariance matrix $\Sigma_0=\Sigma_0(P_0)$. We define $U_n:=n^{1/2}\psi_n$, and note that $U_n$ tends in distribution to $U_0$ provided $P_0\in\mathcal{M}_0$. In this work, the statistic $U_n$ will be used as a basis for the tests we construct. Our primary focus is on applications in which $\psi_n$ is an asymptotically linear estimator of $\psi_0$, in which case $\Sigma_0$ can be characterized in terms of the (multivariate) influence function of $\psi_n$. Below, we will utilize knowledge of this influence function to determine what values of $\psi_n$ are far enough from the zero vector to warrant rejecting the null hypothesise. Often, this task is accomplished by identifying a multivariate region $\Theta_0\subseteq \mathbb{R}^d$ such that the test rejecting $H_0$ if and only if $U_n\in \Theta_0$ has type I error that tends to the nominal type I error  $\alpha\in(0,1)$ as $n\rightarrow\infty$. Provided $\Theta_0$ is a continuity set of $Q_0$, this property is achieved if $\int I\{u\in \Theta_0\}\,dQ_0(u)=\alpha$ whenever $P_0\in\mathcal{M}_0$. There are typically infinitely many choices of $\Theta_0$, and it may be unclear which to select in practice. Instead, for a given norm $\varphi$ on $\mathbb{R}^d$, we propose to search for a univariate region $\Theta_0^{*}\subseteq \mathbb{R}$ such that $
\int I\{\varphi(u)\in\Theta_0^*\}\,dQ_0(u)=\alpha$ whenever $P_0\in\mathcal{M}_0$. Then, an asymptotically calibrated test is defined by rejecting $H_0$ if and only if $\varphi(U_n)\in \Theta_0^*$. Use of the norm $\varphi$ thus allows conversion of the original multivariate problem into a univariate one.
 
In practice, there are many choices for $\varphi$, and as we will see, the norm used plays an important role in determining the performance of the resulting test. As an  example, we consider the $\ell_p$-norm $\varphi_p$ defined as $(z_1,z_2,\ldots,z_p)\mapsto \|z\|_p:=(\,z_1^p+z_2^p+\ldots+z_d^p\,)^{\frac{1}{p}}$ along with regions of the form $\Theta^*_0(r) = [r,\infty)$. The choice $r_0:=\min\{r:\int I\{\|u\|_p\geq  r\}\,dQ_0(u)\leq \alpha\}$ ensures that $\Theta_0^*:=\Theta_0^*(r_0)$ provides a calibrated test, in the sense that the test rejecting $H_0$ if and only if $\|U_n\|_p\in\Theta_0^*$ has asymptotic type I error equal to $\alpha$.
The corresponding $p$-value is given by $\int I\{\|u\|_p \geq \|U_n\|_p\}\,dQ_0(u)$. Different choices of $p$ may yield tests with a different power profile over various alternatives. To explore this phenomenon, we may consider a simple example comparing tests resulting from the choice $p=2$ versus $p=\infty$, the latter corresponding to the maximum norm $(z_1,z_2,\ldots,z_p)\mapsto  \|z\|_\infty = \max\left\{|z_1|,|z_2|, \dots, |z_d|\right\}$.
Figure \ref{fig:figure1} illustrates the behavior of these two tests in the case $d=2$. In Panel A, 100 draws are taken from a multivariate normal distribution $\distv_0$ with mean zero and identity covariance matrix. The red circle and blue square represent the boundary of the region $\Theta_0$ of these two tests constructed using empirical estimates of the $95^{\text{th}}$ percentile of the distributions of $\|\nrvs\|_2$ and $\|\nrvs\|_\infty$, respectively. All observations in Panel A except the five with  largest $\ell_2$-norm are contained within the red circle. Similarly, the blue square contains all observations in Panel A except the five with largest $\ell_\infty$-norm. Observations that fall within the blue shaded region result in rejection of the null hypothesis if the $\ell_2$-norm is used to define the test but not if the $\ell_\infty$-norm is instead used. Conversely, observations that fall in the red shaded region result in rejection of the null hypothesis if the $\ell_\infty$-norm is used to define the test but not if the $\ell_2$-norm is instead used. The same square and circle are redrawn in panels B and C to illustrate the behavior of the test under alternatives in which either (B) $\psi_{10} = 0$ and $\psi_{20} \ne 0$, or (C) $\psi_{10} \neq 0$ and $\psi_{20} \neq 0$. While both constructions of a rejection region result in valid asymptotic type I error control, depending on the alternative, one test will outperform the other in power. In Panel B, shifting each observation in only a single direction has a larger impact on the maximum norm of the observations compared to the $\ell_2$-norm since the maximum norm only considers the largest coordinate.  This is shown by the numerous observations (given by red squares) outside of the blue box (equivalent to rejecting $H_0$) and inside the red circle (equivalent to failing to reject $H_0$). In contrast, there is only a single observation outside the red circle and inside the blue box (given by blue triangles). The converse trend is shown in panel C, where the $\ell_2$-norm performs better because it takes into account both coordinates of the shift. 

\begin{figure}
	\centering
	\includegraphics[width = \linewidth]{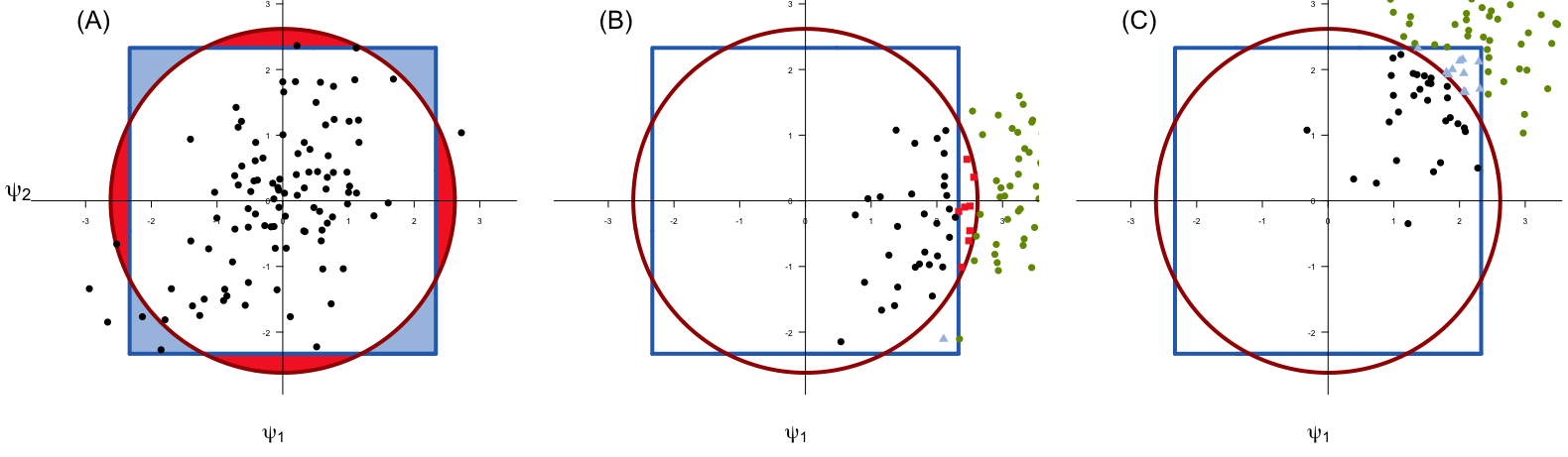}
	\linespread{1}
	\caption{Plots of 100 observations from a limiting distribution of a hypothetical vector of parameter estimators in $\mathbb{R}^2$ (A) under the null, (B) under an alternative with $\psi_1 = 0, \psi_2 \neq 0$, and (C) under an alternative with $\psi_1, \psi_2 \neq 0$. The 95\% quantiles for the data based on the max (blue) and $\ell_2$ (red) norms under the null are given in all three panels. If a test statistic fell within the blue regions the test would fail to reject $H_0$ if the $\ell_\infty$ norm was used, but would reject $H_0$ if the $\ell_2$ norm was used.  The converse is true for the red regions.  Depending on the alternative, the $\ell_\infty$ norm (B) or the $\ell_2$ norm(C) will achieve higher power.}
	\label{fig:figure1}
\end{figure}

\subsection{Adaptive norm selection}

We denote by $\mathscr{F}$ the collection of all norms defined on $\mathbb{R}^d$. So far, we have argued that a test can be defined based on any $\varphi\in \mathscr{F}$ and that the choice of $\varphi$ can influence the power of the test.  In many scenarios, it may not be clear a priori which of several tests should be preferred in a given setting since the power of each test depends on details of the true (unknown) alternative. In order to compare any of several candidate norms, we must first choose an objective criterion for adjudicating, in the setting at hand, the performance of the test statistic $\varphi(U_n)$ for a given norm $\varphi$.

For this purpose, suppose that $\Gamma_0:\mathbb{R}^d\times \mathscr{F} \rightarrow [0,\infty)$ provides a local measure of test inefficiency. Specifically, we stipulate that for any $x \in \mathbb{R}^d \backslash \{0\}$ and $\varphi\in\mathscr{F}$, greater values of $\Gamma_0(x,\varphi)$ indicate a larger asymptotic type II error --- and so, lower power --- for the test based on the test statistic $\varphi(U_n)$ under a location shift by $x$ of the null limiting distribution of $U_n$ under sampling from $P_0$.  In this work, we focus on two particular measures, although our theoretical results are stated in generality. The first, which we refer to as the \emph{acceptance rate} measure, is defined as 
\begin{align}
    \Gamma_{\text{ar}, 0}:(x,\varphi)\mapsto \int \1{\varphi(u+x)\leq c_0}\,dQ_0(u)\ , \label{gamma:accept_rate} 
\end{align}
where $c_0:=\min\{c\geq 0: \int \1{\varphi(u)\leq c}\,dQ_0(u)\geq 1-\alpha\}$ is the smallest cutoff value such that the test rejecting $H_0$ if and only if $\varphi(U_n)>c_0$ has asymptotic type I error equal to $\alpha$. This measure can be interpreted as the asymptotic type II error of the test based on $\varphi(U_n)$ in the context of a sequence of local alternatives under which $\psi_0=\psi^{(n)}_{0}:=x n^{-1/2}$. While it is intuitively simple and straightforward to estimate in practice, this measure can suffer from the fact that its output is constrained in the interval $[0,1 - \alpha]$, so that it becomes less informative --- and thus less useful for discriminating norms --- in settings in which the distribution of $\Gamma_0(U_n, \varphi)$ is concentrated near zero for each norm $\varphi$ considered. Additionally, in view of the exponential tails of the normal distribution, $\Gamma_0(x, \varphi)$ tends to zero rapidly as $x$ tends away from the origin, thereby rendering onerous the task of achieving sufficient relative precision when approximating $\Gamma_0(U_n, \varphi)$ using Monte Carlo methods. These difficulties motivate the consideration of an alternative measure defined as \begin{align}\Gamma_{\text{mf}, 0}:(x,\varphi)\mapsto\min\left\{s\geq 0: \int \1{\varphi(u+sx)\leq c_0}\,dQ_0(u)\leq \tau\right\}\label{eqn:mult_fact_pm}\end{align} for some user-specified $\tau\in(0,1-\alpha)$. We refer to this as the \emph{multiplicative factor} measure since it provides the smallest factor $\kappa$ such that the asymptotic type II error of the test based on $\varphi(U_n)$ is no greater than $\tau$ in the context of a sequence of local alternatives under which $\psi^{(n)}_{0}:=\kappa x n^{-1/2}$. This  measure avoids the drawback of the acceptance rate by operating on  a multiplicative scale, though it does so at the expense of simplicity of interpretation and computational ease.

Suppose that we consider a finite collection $\mathscr{F}_0:=\{\varphi_1,\varphi_2,\ldots,\varphi_K\}\subset \mathscr{F}$ of norms on $\mathbb{R}^d$, which we wish to discriminate based on a given local measure of test inefficiency $\Gamma_0$.  Suppose also that an estimator $\Gamma_n$ of $\Gamma_0$ based on $X_1,X_2,\ldots,X_n$ is available. Then, it is sensible to consider $\Gamma_n(U_n,\varphi)$ as an estimated local measure of test inefficiency for a given norm $\varphi$, where local here refers to consideration of local alternatives defined by $U_n$ itself. As a first attempt at developing a test based on adaptive norm selection, we could consider using the test statistic $\varphi_{k_n(U_n)}(U_n)$ with $k_n(U_n):=\argmin_{k}\Gamma_{n}(U_n,\varphi_k)$ --- this amounts to considering the univariate summary  $\varphi(U_n)$ based on the norm $\varphi\in\mathscr{F}_0$ with the smallest estimated local measure of test inefficiency. However, the test statistic $\varphi_{k_n(U_n)}(U_n)$ appears difficult to make valid inference with since its limit distribution is difficult to derive --- for example, the lack of continuity of  $\varphi_{k_n(U_n)}(U_n)$ as a function of $U_n$ precludes the use of a continuous mapping theorem. More importantly, this test statistic produces an undesirable ordering in the space of alternatives, as illustrated in Figure \ref{fig:adapt_norm_issue} with a simple example in which $\mathscr{F}_0 = \{\varphi_1, \varphi_2\}$ contains only two given norms, and the parameter of interest takes values in $\mathbb{R}^{2}$. In the figure, for each alternative, the color indicates which of $\varphi_1$ (dark red) and $\varphi_2$ (light blue) is preferred to define a test.  However, the norm $\varphi_1$ takes values that are substantially larger than $\varphi_2$ for observations that are similar. As a result for the two points $u_{n, 1}$, $u_{n, 2}$, even though $\varphi_1(u_{n, 1}) < \varphi_1(u_{n, 2})$ and $\varphi_2(u_{n, 1}) < \varphi_2(u_{n, 2})$, it is also true that  $\varphi_{k_n(u_{n, 2})}(u_{n, 2}) < \varphi_{k_n(u_{n, 1})}(u_{n, 1})$.  Thus, even though $u_{n, 2}$ is further away from the null as measured by both norms, the relative size of the norms makes it appear as though $u_{n, 1}$ is more surprising than $u_{n, 2}$ under the null according to the adaptive norm test statistic.

\begin{figure}
		\centering
		\linespread{1}
\includegraphics[width = \linewidth]{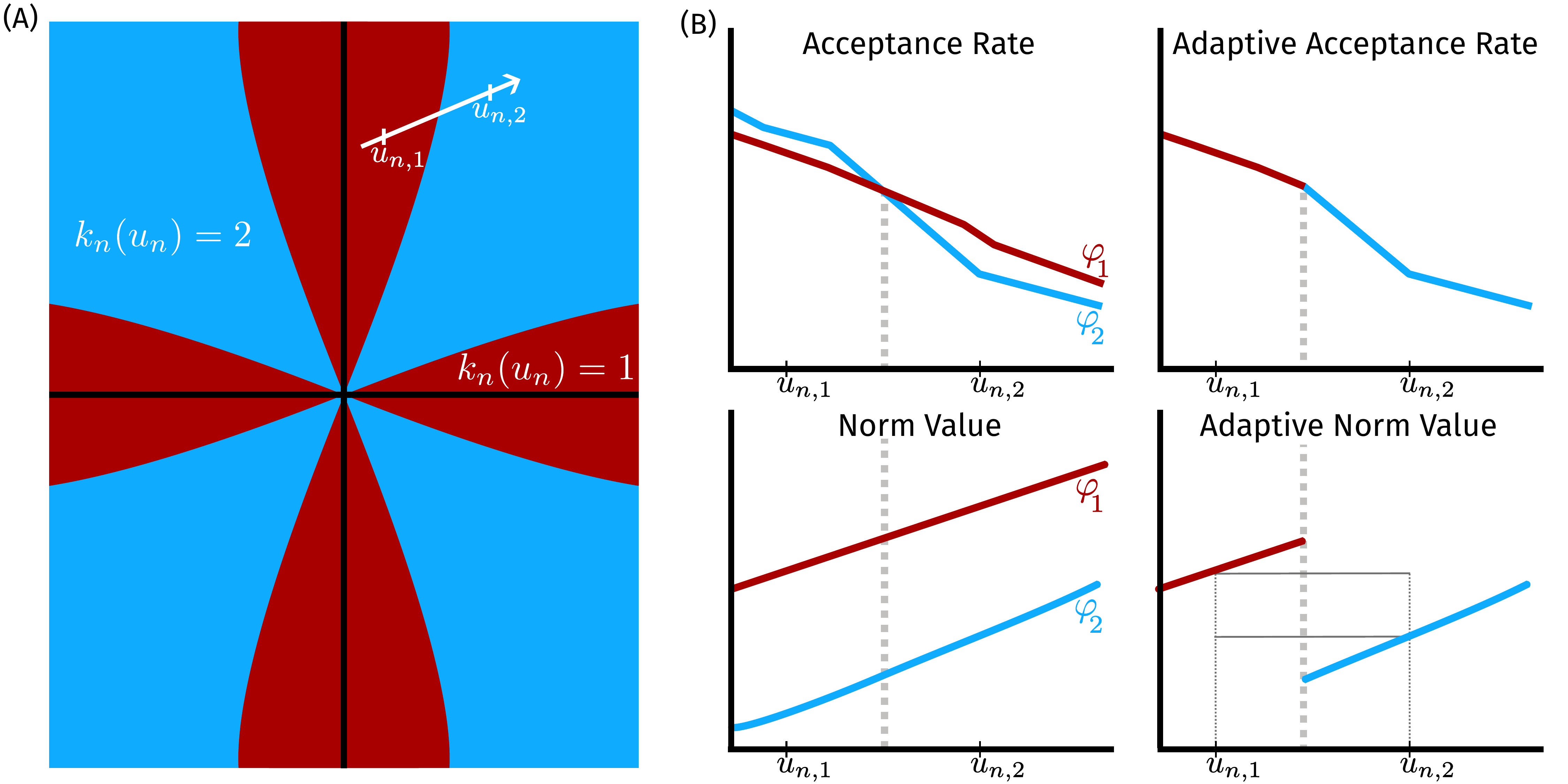}
\caption{This figure illustrates two issues that could arise when using the adaptive norm value as a test statistic. (A) shows regions of $\mathbb{R}^{2}$ in which $\varphi_1$ (dark red) or $\varphi_2$ (light blue) have better (hypothetical) acceptance rate value.  A line  segment containing two points $u_{n, 1}$ and $u_{n, 2}$ is also shown, and the points along this line segment form the $x$-axis of the four figures in (B). The arrow indicates the direction along the line segment in which both $\varphi_1$ and $\varphi_2$ increase. The top left display in (B) shows the hypothetical values of the acceptance rate measure along the $u_n$-values shown on the white arrow in (A), and the top right panel shows the adaptive version of this measure (the pointwise minimum of the individual acceptance rate measures).  The bottom left display indicates the norm values, and the bottom right display shows the adaptive norm value wherein the norm with lowest acceptance rate is used.  As shown by the two horizontal line segments in this display, the adaptive norm value does not necessarily increase as $u_n$-values are taken further away from the origin. Additionally, the discontinuity of the adaptive norm value is apparent in the bottom right display.}
	\label{fig:adapt_norm_issue}
\end{figure}

In view of these challenges, we consider another strategy for building an aggregate test statistic.  We observe that for any norm $\varphi$, if $\Gamma_0$ were known,  $\Gamma_0(U_n,\varphi)$ could serve as a sensible alternative to the test statistic $\varphi(U_n)$, with smaller values of $\Gamma_0(U_n,\varphi)$ supporting rejection of the null hypothesis. The use of $\Gamma_0(U_n, \varphi)$  as a test statistic has desirable properties.
First, the interpretation of the realizations of $\Gamma_0(U_n, \varphi)$ depends neither on the norm used nor on the limiting distribution $Q_0$. As a result, the value of $\Gamma_0(u_n, \varphi)$ can be directly compared across choices of $\varphi$, thereby facilitating the construction of an adaptive test statistic. Second, the ordering induced on the parameter space by $\Gamma_0$ is sensible. To illustrate this, suppose first that two realizations  $u_{n,1}$ and $u_{n,2}$ of $U_n$ fall on a common ray, that is, $u_{n,1}=\beta_1 v$ and $u_{n,2}=\beta_2 v$ for some direction $v\in\mathbb{R}^d$ and non-negative values $\beta_1,\beta_2\in\mathbb{R}$.  For $\beta_1>\beta_2$, we expect $u_{n,1}$ to be less likely than $u_{n,2}$ under the null, and indeed, in these settings $\Gamma_0(u_{n,1},\varphi)<\Gamma_0(u_{n,2},\varphi)$ under regularity conditions introduced in the next section. 
Suppose instead that $u_{n,1}$ and $u_{n,2}$ are such that $\varphi(u_{n,1}) = \varphi(u_{n,2})$. In this case, depending on $Q_0$, either $u_{n,1}$ or $u_{n,2}$ could be more likely. A test statistic based on $\Gamma_0$ allows for consideration of $Q_0$ and thus permits differentiation of $u_{n,1}$ and $u_{n,2}$ even when these realizations may be undifferentiated by $\varphi$.

In practice, since $\Gamma_0$ is unknown, the test statistic $\Gamma_n(U_n,\varphi)$ would be used instead of $\Gamma_0(U_n,\varphi)$. We observe that, by definition, \begin{equation}
    Z_n:=\Gamma_{n}(U_n, \varphi_{k_n(U_n)})=\min\{\Gamma_{n}(U_n, \varphi_1),\Gamma_{n}(U_n, \varphi_2), \ldots, \Gamma_{n}(U_n, \varphi_K)\}\ .\label{teststat}
\end{equation} As such, the adaptive test statistic $Z_n$ is a continuous transformation of the single-norm statistics $\Gamma_{n}(U_n,\varphi_1),\Gamma_{n}(U_n,\varphi_2),\ldots,\Gamma_{n}(U_n,\varphi_K)$. Provided $u\mapsto \Gamma_0(u,\varphi)$ is continuous for each $\varphi\in \mathscr{F}_0$, this implies that a non-degenerate limit distribution can be derived for the test statistic $Z_n$, thereby facilitating valid inference. Specifically, under regularity conditions, we may expect $Z_n$ to converge in distribution to the random variable $Z_0:=\min\{\Gamma_{0}(U_0,\varphi_1),\Gamma_{0}(U_0,\varphi_2),\ldots,\Gamma_{0}(U_0,\varphi_K)\}$, where $U_0$ is distributed according to $Q_0$. This motivates an adaptive test in which we
reject $H_0$ if and only if $Z_n>\chi_n$, where $\chi_n$ is any consistent estimator of the $(1-\alpha)$-quantile $\chi_0$ of the distribution of $Z_0$.  

We note here that our proposed procedure was inspired by the proposal of \citet{zhang_comment_2015}, which can be considered a special case of our framework. Their method, which focuses on the specific problem of testing for null correlations in the setting of univariable linear models, is recovered by taking $\psi_n$ to be the vector of sample correlations, $\mathscr{F}_0$ to be a collection of sum-of-squares norms, and $\Gamma_0$ to be the observed $p$-value for the test based on $\varphi(U_n)$. The sum-of-squares norm is defined as \[\jmath_{k}:(z_1,z_2,\ldots,z_d)\mapsto \left\{\textstyle\sum_{j=1}^{k}z^2_{(d-j+1)}\right\}^{1/2}\]for any fixed $k\in\{1,2,\ldots,d\}$ and with $z^2_{(r)}$ denoting the $r^{th}$ order statistic based on $z^2_1,z_2^2, \ldots, z^2_d$ for each $r=1,2\ldots,d$. A proof that the sum-of-squares norm is indeed a proper norm is provided in Lemma \ref{lemma:ssq_norm} of the Supplement.

\subsection{Implementation of proposed adaptive test}
\label{ssec:obtaining_null}

Suppose that $\Sigma_n$ is a consistent estimator of $\Sigma_0$, and denote by $Q_n$ the distribution function of the normal distribution with mean zero and covariance matrix $\Sigma_n$. An estimator $\Gamma_{n}$ can be derived by replacing $Q_0$ by $Q_n$ in the definition of $\Gamma_0$. We set $Z^{*}_n:= \min\{\Gamma_{n}(\bar{U}_n,\varphi_1),\Gamma_{n}(\bar{U}_n,\varphi_2)\ldots,\Gamma_{n}(\bar{U}_n,\varphi_K)\}$, where $\bar{U}_n$ represents a random draw from $Q_n$, and note that $Z_n^*$ serves as a natural proxy for a random draw from the null limit distribution of $Z_n$. Because the distribution of $Z_n$ is difficult to calculate in practice, we instead define our cutoff value $\chi^*_n$ as the $(1 - \alpha)$-quantile of $Z_n^{*}$. Below, we will establish properties of the test in which we
\begin{equation}\label{test}
    \text{reject }H_0\text{ if and only if }Z_n>\chi_n^*\ .
\end{equation}
While an analytic form $\chi^*_n$ is not currently available, its value can be approximated with an arbitrary level of accuracy using the following steps:
\begin{enumerate}
    \item for $M$ large, conditionally on $Q_n$, generate independent draws $\bar{U}_{n,1},\bar{U}_{n,2},\ldots,\bar{U}_{n,M}$ from $Q_n$;
    \item set $\bar{Z}_{n,m}:=\min\{\Gamma_{n}(\bar{U}_{n,m},\varphi_1),\Gamma_{n}(\bar{U}_{n,m},\varphi_2),\ldots,\Gamma_{n}(\bar{U}_{n,m},\varphi_K)\}$ for $m=1,2,\ldots,M$;
    \item compute the sample $(1-\alpha)$-quantile $\chi^*_{n,m}$ based on $\{\bar{Z}_{n,1},\bar{Z}_{n,2},\ldots,\bar{Z}_{n,M}\}$.
\end{enumerate}For concreteness of discussion, suppose that, for each $j=1,2,\ldots,K$, $\psi_{jn}$ is an asymptotically linear estimator of $\psi_{j0}$ with influence function $\phi_j:\mathcal{X}\rightarrow\mathbb{R}$, in the sense that \[\psi_{jn}=\psi_{j0}+\frac{1}{n}\sum_{i=1}^{n}\phi_j(X_i)+o_P(n^{-1/2})\]with $E_0\{\phi_j(X)\}=0$ and $var_0\{\phi_j(X)\}<\infty$. Suppose that the form of each $\phi_j$ is known up to some dependence on the unknown data-generating distribution $P_0$. Asymptotic linearity of $\psi_{n}$ readily implies that, under the null hypothesis, $U_n$ tends to a random vector following a multivariate normal distribution with mean zero and covariance matrix $\Sigma_0$ with $jk^{th}$ element $\Sigma_{jk0}:=\int \phi_j(x)\phi_k(x)\,dP_0(x)$. We will require a consistent estimator $\Sigma_n$ of $\Sigma_0$ in our developments --- a natural candidate is the empirical cross-moment estimator,  defined entrywise as $\Sigma_{jkn}:=\frac{1}{n}\sum_{i=1}^{n}\phi_{jn}(X_i)\phi_{kn}(X_i)$, where $\phi_{jn}$ and $\phi_{kn}$ are estimators of the influence functions $\phi_j$ and $\phi_k$. While for simplicity this empirical estimator is employed in all simulations and data analyses reported below, more sophisticated procedures for covariance estimation --- e.g., as described by \cite{ledoit_well-conditioned_2004, ledoit_power_2020} --- could be used instead. The implementation of our approach also requires selection of a collection $\mathscr{F}_0$ of norms. In this article, we explicitly consider the $\ell_p$ and sum-of-squares norms.

\section{Large-sample properties of proposed test}  \label{sec:theorems}

In this section, we establish conditions under which the adaptive test outlined in \eqref{test} is guaranteed to have desirable statistical properties. In addition to type I error control and consistency against fixed alternatives, we will show that our proposed test has nontrivial power against local alternatives. For each theorem in this Section, a proof is provided in the Supplement.

Since $\Gamma_0$ depends on $Q_0$ only through $\Sigma_0$, we explicitly denote the local measure of test inefficiency as a fixed mapping $(u,\Sigma,\varphi)\mapsto\Gamma(u,\Sigma,\varphi)$ for which we have that $\Gamma_0(u,\varphi)=\Gamma(u,\Sigma_0,\varphi)$ for each $u\in\mathbb{R}^d$ and $\varphi\in\mathscr{F}_0$. For simplicity, we consider an arbitrary estimator $\Gamma_n$ of $\Gamma_0$ of the form   $\Gamma_n(u,\varphi)=\Gamma(u,\Sigma_n,\varphi)$ for each $u\in\mathbb{R}^d$ and $\varphi\in\mathscr{F}_0$, where $\Sigma_n$ is any consistent estimator of $\Sigma_0$.
By explicitly representing the dependence of $\Gamma_0$ and $\Gamma_n$ on $\Sigma_0$ and $\Sigma_n$, respectively, via $\Gamma$, the consistency of $\Gamma_n$ to $\Gamma_0$ can be established as a consequence of a simple continuity condition on $\Gamma$. We introduce the following conditions on the local measure of test inefficiency relative to a given norm $\varphi$, where we denote by $\mathbb{V}_d$ the space of all positive definite $d\times d$ matrices:

\begin{enumerate}[label= C\arabic{enumi}), ref = C\arabic{enumi}]
	\item $(u,\Sigma)\mapsto \Gamma(u,\Sigma,\varphi)$ is continuous and non-negative on $\mathbb{R}^d \times B_0$ for some neighborhood $B_0\subset \mathbb{V}_d$  of $\Sigma_0$; \label{itm:cond_cont_non_neg}
	\item $\int I\{\Gamma_0(u,\varphi)=t\}\,dQ_0(u)=0$ for every $t\geq 0$;\label{itm:cond_dom_by_lebesgue}
	\item $\Gamma(x_s, \Sigma, \varphi)\rightarrow 0$ uniformly over $\Sigma \in B_1$ for some neighborhood $B_1\subset \mathbb{V}_d$ of $\Sigma_0$ for every sequence $x_1,x_2,\ldots$ of elements of $\mathbb{R}^d$ such that $\varphi(x_s)\rightarrow \infty$; \label{itm:cond_conv_prb_zr} 
	\item $u\mapsto \Gamma_0(u, \varphi)$ is quasi-concave, in the sense that $\{u : \Gamma_0(u,\varphi) \geq a\}$ is convex for every $a\geq 0$; \label{itm:cond_unimod}
	\item $u\mapsto \Gamma_0(u,\varphi)$ is centrally symmetric, in the sense that $\Gamma_0(u,\varphi)=\Gamma_0(-u,\varphi)$ for every $u\in\mathbb{R}^d$. \label{itm:cond_cent_sym}
\end{enumerate}

The result below states that, under mild conditions, the proposed  test has valid type I error rate and power tending to one under each fixed alternative as sample size tends to infinity.

\begin{theorem}
	\label{thm:t1ec}
	Suppose that 
	conditions \ref{itm:cond_cont_non_neg}--\ref{itm:cond_dom_by_lebesgue} hold for each $\varphi\in\mathscr{F}_0$. Under sampling from $P_0$, as $n\rightarrow\infty$, the rejection rate of the proposed test \eqref{test}:
	\begin{enumerate}[label=\alph*)]
	    \item tends to $\alpha$  if $P_0\in\mathscr{M}_0$;
	    \item tends to 1 if $P_0\notin \mathscr{M}_0$ provided condition \ref{itm:cond_conv_prb_zr} also holds for some $\varphi\in\mathscr{F}_0$.
	\end{enumerate}
\end{theorem}	
Since in practice studies are typically designed to have power substantively below one in view of cost and other logistic constraints, studying the asymptotic behavior of the proposed test for these settings is of interest and motivates consideration of local alternatives.
Specifically, a local alternative to $P_0\in\mathscr{M}_0$ is a one-dimensional parametric submodel  $\{P_t\}\subset \mathscr{M}$  of $\mathscr{M}$ dominated by $P_0$ and such that the Radon-Nikodym derivative of $P_t$ relative to $P_0$ satisfies, for $t$ in a neighborhood of zero, \begin{align}\label{local}
\frac{dP_t}{dP_0}(x)=1+tg(x)+tr_t(x)
\end{align}for some element  $g$ in the tangent space of $\mathscr{M}$ but not in the tangent space of $\mathscr{M}_0$ at $P_0$, and where $r_t$ is a remainder term tending to zero in a uniform sense \citep{pfanzagl_estimation_1990}. The estimator $\psi_n$ is said to be regular at $P_0\in\mathscr{M}$  if the limit distribution of $n^{1/2}(\psi_n-\psi_0)$ under sampling from $P_0$ and of $n^{1/2}(\psi_n-\psi_0^{(n)})$ under sampling from $P_0^{(n)}:=\left.P_t\right|_{t=n^{-1/2}}$ is the same, where we write $\psi_0^{(n)}:=\Psi(P_0^{(n)})$ and $\{P_t\}$ is any local alternative to $P_0$.  We note that, for any such sequence $P_0^{(n)}$, it holds that $\psi_0^{(n)}=a n^{-1/2}+o(1)$ for some $a\in\mathbb{R}^d\setminus \{0\}$. The following theorem states that, under certain regularity conditions, if the estimator $\psi_n$ is regular, then the proposed test is locally unbiased in the sense that it has non-trivial power under local alternatives.

\begin{theorem}
	\label{thm:unbiased_locl_alt}
Suppose that $P_0\in\mathscr{M}_0$, and let $P_0^{(n)}$ be a sequence of local alternatives converging to $P_0$.  Suppose also that conditions \ref{itm:cond_cont_non_neg}, \ref{itm:cond_dom_by_lebesgue}, \ref{itm:cond_unimod} and \ref{itm:cond_cent_sym} hold for each $\varphi\in\mathscr{F}_0$, that condition \ref{itm:cond_conv_prb_zr} holds for some $\varphi\in\mathscr{F}_0$,
and that $\psi_n$ is a regular estimator of $\psi_0$ under sampling from $P_0$. Then, the rejection rate $\pi_n$ of the proposed test under sampling from $P_0^{(n)}$ satisfies that $\liminf_n \pi_n >\alpha$.
\end{theorem}This theorem guarantees that the rejection rate is greater under local alternatives than it is under the null.  Theorems \ref{thm:t1ec} and \ref{thm:unbiased_locl_alt} indicate that the proposed test has desirable properties provided several conditions on the local measure of test inefficiency used hold. The next result establishes that the two measures presented in Section 3, namely the acceptance rate and multiplicative factor measures, indeed satisfy all required conditions, and therefore, can be used in our procedure.
\begin{theorem}
	\label{thm:performance_metric}
	Both the acceptance rate measure \eqref{gamma:accept_rate} and the multiplicative factor measure \eqref{eqn:mult_fact_pm} satisfy conditions \ref{itm:cond_cont_non_neg}--\ref{itm:cond_cent_sym} for each norm $\varphi$.
\end{theorem}

\section{Numerical examples}
\label{sec:sim_stdy}

In this section, we discuss the implementation and evaluate the performance of our proposed test in the context of the three working examples introduced in Section \ref{sec:Working Examples}.

In each example, we consider all combinations of sample size $n\in\{ 100, 200,500\}$ and covariate vector dimension $d\in\{10, 50,100\}$. The multiplicative factor measure  \eqref{eqn:mult_fact_pm} is used throughout. We compare a variety of competing procedures, including adaptive and non-adaptive versions of our test. The non-adaptive tests use the $\ell_2$ and maximum absolute value norms, and are referred to  as the $\ell_2$ and $\ell_\infty$ tests, respectively.  The first adaptive version of our test selects over the $\ell_1$, $\ell_2$, $\ell_4$, $\ell_6$ and $\ell_\infty$ norms, and is referred to as the adaptive $\ell_p$ test. The second adaptive test selects over various versions of the $\jmath_{k}$ norm --- specifically, over $k\in\{1, 3, 5, 6, 8, 10\}$ when $d = 10$,  $k\in\{1, 11, 21, 30, 40, 50\}$ when $d = 50$, and $k\in\{1, 21, 41, 60, 80, 100\}$ when $d = 100$ --- and is referred to as the sum-of-squares test.  We note that $\jmath_1 = \ell_\infty$ and $\jmath_d = \ell_2$.  We contrast the performance of these adaptive procedures with two existing all-purpose methods for multiple testing. Each all-purpose method (including ours) uses the same covariance matrix estimator $\Sigma_n$ and parameter estimator $\psi_n$. The first is a test based on the Bonferroni-corrected $p$-value $d\times \min(p_{1n},p_{2n},\ldots,p_{dn})$ computed from individual $p$-values $p_{jn}:=2\,[1-\Phi(z_j)]$, where $z_j:=n^{1/2}|\psi_{nj}|/\sigma_{jn}$, $\sigma_{jn}$ is a consistent estimator of the asymptotic standard deviation $\sigma_{0j}$ of $n^{1/2}\psi_{jn}$ under the null hypothesis, and $\Phi$ represents the standard normal distribution function. In our simulations, we take $\sigma_{jn}$ to be the root of the empirical second moment of $\phi_{nj}(X_1),\phi_{nj}(X_2),\ldots,\phi_{nj}(X_n)$, where $\phi_{nj}$ is a consistent plug-in estimator of the influence function $\phi_{0j}$ of $\psi_{nj}$. The second is the more recent Cauchy combination test (referred to here as the Cauchy test) described by \cite{liu_cauchy_2020} based on the test statistic $d^{-1  }\sum_{j = 1}^d \text{tan}\{( 2p_{jn} - 3/2 ) \pi\} $. Under certain conditions, including mutual independence of $\psi_{n1},\psi_{n2}, \ldots,\psi_{nd}$, this test statistic has a limiting Cauchy distribution under the null hypothesis. However, \cite{liu_cauchy_2020} show that even when independence fails to hold, $p$-values computed using the Cauchy distribution are approximately valid for large realizations of the test statistic.  For this reason, and in view of its simplicity, we include this test as a comparator in our simulation studies.  

\subsection{Example 1: correlation}
\label{sec:cor_examp}

In this example, we consider the settings described in the first example of \cite{mckeague_adaptive_2015} and \cite{zhang_comment_2015}.  The vector $W = (W_1, W_2,\dots, W_d)$ of covariates  is generated from a normal distribution with mean zero and  covariance matrix with diagonal and off-diagonal terms equal to 1 and  $\rho$, respectively. Three distinct conditional outcome distributions are considered. In each setting, we generate $\varepsilon$ as a standard normal variable independent of $W$. Conditionally on $W$ and $\epsilon$, we separately consider \begin{enumerate}\setlength{\itemindent}{4em}
    \item[(Setting 1)] $Y = \varepsilon$;\vspace{-.1in}
    \item[(Setting 2)] $Y = 0.25W_1 + \varepsilon$;\vspace{-.1in}
    \item[(Setting 3)] $Y = 0.15\,(W_1+\ldots+W_5)-0.1\,(W_6+\ldots+W_{10})+ \varepsilon$.
\end{enumerate}

In this example, the sampling distribution of each test statistic --- and thus cutoffs upon which to construct valid tests --- can also be determined using two different methods.  The standard approach, discussed above and referred to as the parametric bootstrap test, estimates the limiting distribution of $n^{1/2}(\psi_n-\psi_0)$ using a mean-zero normal distribution with estimated covariance matrix. A permutation approach, which typically provides better calibration than the parametric bootstrap in smaller samples, can also be used in this particular example. A permutation-based approximation of the $p$-value of the test based on $Z_n$ can be obtained by independently generating modifications of the original dataset in which the outcome vector has been randomly permuted across observations, re-computing $Z_n$ for each such permuted dataset, and computing the fraction of permuted datasets for which the re-computed $Z_n$ value is larger than the original $Z_n$ value. In this simulation, permutation-based and parametric bootstrap versions of our adaptive test were compared to three competing tests, namely the test of \citet{zhang_comment_2015}, the Bonferroni test, and the Cauchy test. We note here that the test of \citet{zhang_comment_2015} leverages knowledge about the data-generating mechanism, whereas other procedures considered instead make use of nonparametric parameter and covariance estimators.

\begin{figure}
		\centering
		\linespread{1}
\includegraphics[width = \linewidth]{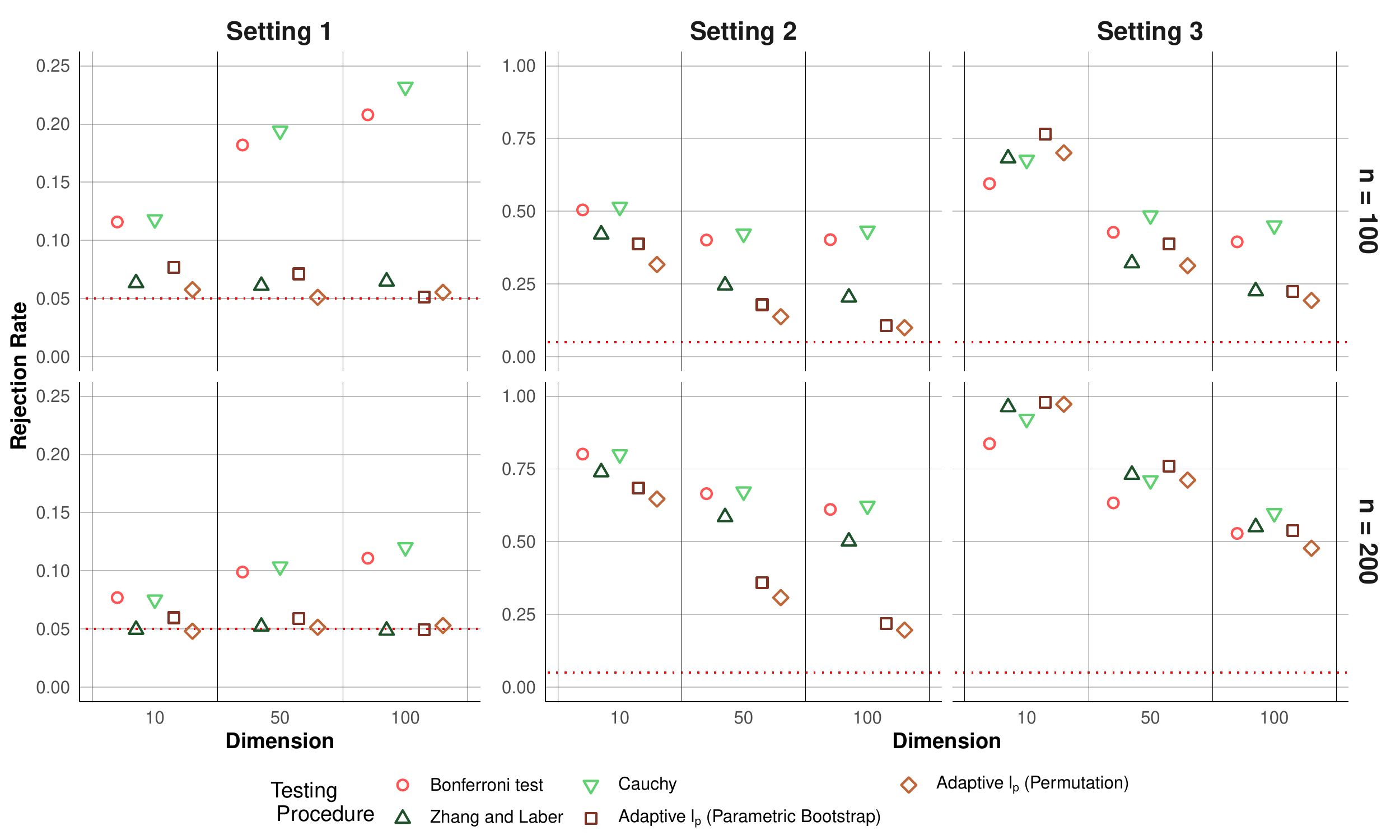}
\caption{Empirical rejection rate of various tests applicable in Example 1 under different data-generating mechanisms, at different sample sizes, and for covariate vectors with no correlation across components and of different length.}
	\label{fig:no_cor_ex1}
\end{figure}
	
\begin{figure}
\centering
\includegraphics[width = \linewidth]{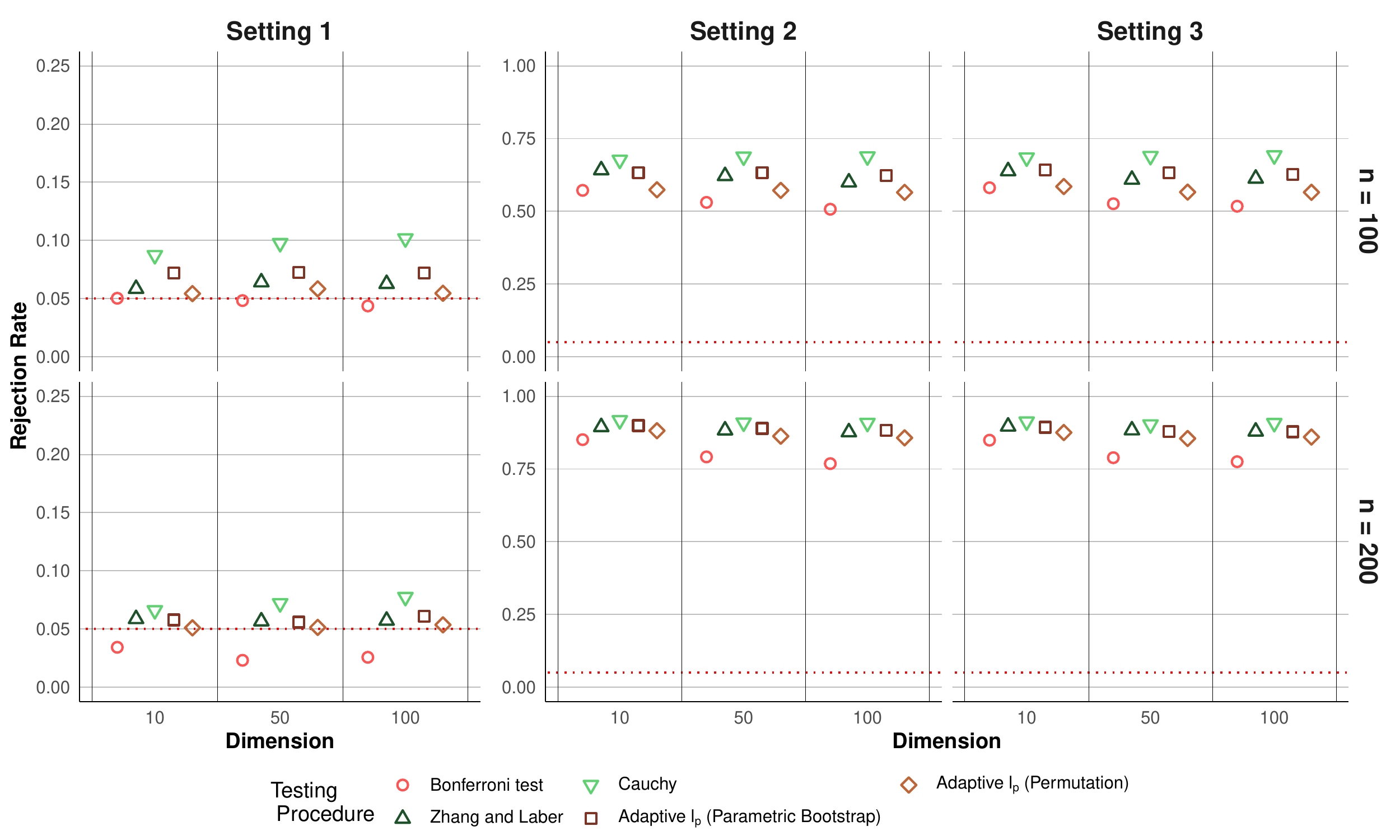}
	\caption{Empirical rejection rate of various tests applicable in Example 1 under different data-generating mechanisms, at different sample sizes, and for covariate vectors with high correlation (80\%) across components and of different length.}
	\label{fig:much_cor}
\end{figure}

The empirical rejection rates of the different tests considered are shown for the three described settings in Figures \ref{fig:no_cor_ex1} and \ref{fig:much_cor} with $\rho=0$ and $\rho=0.8$, respectively. Results for the intermediate setting $\rho = 0.5$ are provided in Figure \ref{fig:some_cor} in the Supplement. Results for  $n=500$ are not shown because power is very close to one for Settings 2 and 3.  Because $H_0$ holds in Setting 1, we expect the rejection rates for this setting to be close to the nominal level $0.05$.  Figures \ref{fig:no_cor_ex1} and \ref{fig:much_cor} illustrate that this is achieved by every testing procedure evaluated except for the Bonferroni test and the Cauchy test. In Settings 2 and 3, $H_0$ does not hold and the plots convey the empirical power of the tests considered.   

In most scenarios in which $H_0$ does not hold all tests have similar empirical power. In most settings, the test proposed by \citet{zhang_comment_2015} slightly outperforms all other tests, and in settings in which it is not the best, it only performs slightly worse than the best test. The most noticeable differences in performance are found in Setting 2 at sample size $n = 200$ under mutual independence of covariate component in which the \cite{zhang_comment_2015} test outperforms all others by a substantial margin.  This superior performance is expected since the data-generating mechanism matches the model assumed in this method, whereas the other tests considered are nonparametric and therefore valid under weaker conditions. 

Both of our adaptive tests perform similarly, with the permutation-based test having lower power but achieving better type I error control than the parametric bootstrap-based test. The relatively higher empirical power of the parametric bootstrap test relative to the permutation-based test likely stems from the fact that the null hypothesis tested is weaker for the latter (null marginal associations) than for the former (joint independence), and that the parametric bootstrap test is imperfectly calibrated, as evidenced by its slightly inflated type I error.

In this simulation, the Bonferroni test and the Cauchy test are anti-conservative, especially in settings in which there is no correlation between covariates. The failure of these tests to achieve nominal type I error control can  mostly be attributed to difficulty estimating the variance of $\psi_n$. Figure \ref{fig:pval_all} in the Supplement shows the distribution of $p_{1n}$ in the setting in which there is no correlation between covariate components.
This distribution has a large spike near zero, although the spike is less pronounced at sample size $n=200$ and is expected to dissipate as sample size further increases. We observe that when the $p$-value $p_{1n}$ is computed using the true standard error of $\psi_n$, this spike vanishes, suggesting that the poor small-sample calibration stems from estimation of the standard error. The observed over-representation of small $p$-values (relative to the uniform distribution) causes a large inflation in type I error for the Bonferroni test. The $p$-value of the Cauchy test is also sensitive to small values of $p_{jn}$ due to the vertical asymptote of the tangent function used to define the test statistic.

\subsection*{Example 2: coefficients of a working log-linear regression model under missingness}

In the second example, conditionally on $W=w$, the binary outcome $Y$ is simulated from the logistic regression model $\pr(Y = 1 \mid W=w) =\textrm{expit}( \beta_1w_1+\beta_2w_2+\ldots+\beta_dw_d)$.  In all scenarios, the conditional missingness probability is given by $ \pr(\Delta = 1 \mid  W=w) = \textrm{expit}(0.5 + 0.15 w_{d - 1} - 0.275 w_d)$, and the vector $W = (W_1,W_2, \dots, W_d)$ of covariates is drawn from a multivariate normal with mean zero and covariance matrix with diagonal and off-diagonal entries equal to 1 and 0.5, respectively.  We separately consider the following settings defined by different values for the regression coefficient vector:  \begin{enumerate}\setlength{\itemindent}{4em}
    \item[(Setting 1)] $\beta_1=\ldots=\beta_d=0$;\vspace{-.1in}
    \item[(Setting 2)] $\beta_1=0.6$, $\beta_2=\ldots=\beta_d=0$;\vspace{-.1in}
    \item[(Setting 3)] $\beta_1=\ldots=\beta_5=0.32$, $\beta_6=\ldots=\beta_{10}=-0.32$, $\beta_{11}=\ldots=\beta_d=0$;\vspace{-.1in}
    \item[(Setting 4)] $\beta_1=\ldots=\beta_5=0.23375$, $\beta_6=\ldots=\beta_{10}=0.4675$, $\beta_{11}=\ldots=\beta_d=0$.
\end{enumerate}Thus, the null hypothesis holds in Setting 1 but not in any of Settings 2, 3 and 4. In this example, an influence function-based estimator of the covariance matrix $\Sigma_0$ was used, and conditional mean functions involved were estimated using either an elastic net \citep{simon_blockwise_2013,tibshirani_strong_2012,friedman_regularization_2010, tibshirani_strong_2012} or loess smoother.

In the null setting (Setting 1), we find that the type one error of all tests is near (though still slightly above) the $0.05$ type one error rate. In general, the type one error is higher in settings with smaller sample size and larger dimension (as expected). In Setting 2, all tests have similar power with the Cauchy test slightly outperforming and the Bonferroni test slightly under-performing all other tests.  The differences in performance are larger  for settings with higher dimension.  In Setting 3, the $\ell_\infty$ based test  outperforms all others, especially in the sample size 500 setting.  In Setting 4, all tests except the $\ell_\infty$ and Bonferroni test perform nearly identically well for each sample size and dimension. In Setting 3, ten covariates are associated with the outcome, which would suggest that norms accounting for the many non-null associations would perform (relatively) better, as seen in Figure \ref{fig:figure1}. Unexpectedly, the $\ell_\infty$ test had the largest power. This finding may be driven by the fact that other norms place larger importance on smaller component values. In this setting, while only ten covariates are directly associated with the outcome, all other covariates are still marginally associated with the outcome through their correlation with other covariates.  While the $\ell_\infty$ norm considers only the covariate most strongly associated with the outcome, other norms consider all covariates.  Covariates that are indirectly associated with the outcome thus have a small (though still non-zero) association with the outcome.  If the additional variability introduced by including these covariates is too large, it may be detrimental to test performance. This explanation is supported by results presented in Figure \ref{fig:dat_examp_2}, wherein we find that when all covariates are truly associated with the outcome ($d = 10$), all considered tests have comparable power, but that differences emerge in larger dimensions.  It may also be that the low power of the adaptive and $\ell_2$ tests are a consequence of various linear effects on the outcome canceling each other out. Because covariate vector components are highly correlated and there are an equal number of positive and negative $\beta$ values of the same magnitude, the combined effect from all covariates could be small.  This would make it more difficult to discern the marginal effect of any single covariate. In Setting 4, ten covariates are directly associated with the outcome, just like in Setting 3, though unlike Setting 3, all non-null regression coefficients are positive. These differences result in a reversal of which tests are optimal, with the $\ell_\infty$ test having the lowest power of the non-adaptive tests, and the adaptive tests and the $\ell_2$-based test all perform nearly equally, with the latter narrowly outperforming the former.  

Overall, we see that depending on the scenario, the norm on which a test is based could be unimportant (Setting 2) or a source of substantial differences between tests (Setting 3).  In settings in which the choice of the norm is consequential, the adaptive test does not outperform all fixed norm tests, but does provide consistent performance across all settings. This example also suggests that common guidelines from the high-dimensional statistics literature on which norm should perform best in a given scenario may not be reliable, even when the data-generating mechanism is known a priori.

\begin{figure}
	\centering
	\linespread{1}
\includegraphics[width = \linewidth]{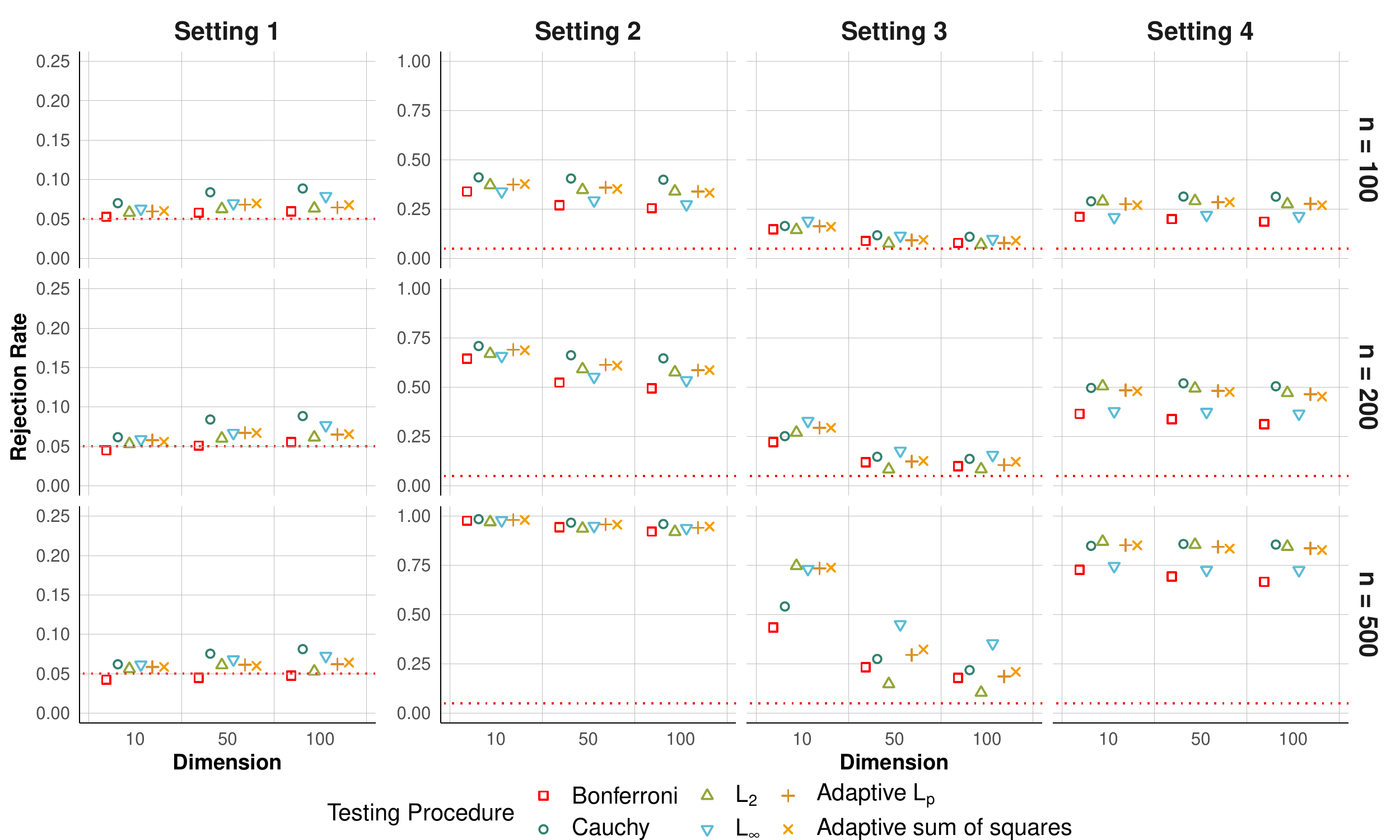}
	\caption{Empirical rejection rate of various tests applicable in Example 2 under different data-
generating mechanisms, at different sample sizes, and for covariate vectors with moderate correlation (50\%)
across components and of different length.}
\label{fig:dat_examp_2}
\end{figure}

\subsection*{Example 3: coefficients of a working effect modification model for randomized trials}

In this example, the covariate vector $W$ is drawn from a multivariate normal distribution with mean zero and covariance matrix with diagonal and off-diagonal entries equal to 1 and 0.5, respectively. Given $W=w$, the binary exposure $A\in\{0,1\}$ is drawn from a binomial distribution with success probability 0.5, as in a standard randomized trial. Finally, given $(A,W)=(a,w)$, the binary outcome $Y\in\{0,1\}$ is drawn from a Bernoulli distribution with success probability given by \[\logit\pr\left(Y = 1 \mid W=w, A=a\right) = \beta_0 a + \sum_{j = 1}^{d} (\beta_{j}+\gamma_ja) w_j \ .\]We set $\beta_0 = 0.2$, $\beta_1= \ldots=\beta_{d/2} = 0.7 / d$ and $\beta_{d/2+1}=\ldots=\beta_{d} = 0$, where $d$ is even, and consider the following settings: \begin{enumerate}\setlength{\itemindent}{4em}
    \item[(Setting 1)] $\gamma_{1}=\ldots=\gamma_{d} = 0$;\vspace{-.1in}
    \item[(Setting 2)] $\gamma_{1} = 1.2$ and $\gamma_{2}=\ldots=\gamma_{d} = 0$;\vspace{-.1in}
    \item[(Setting 3)] $\gamma_{1}=\ldots=\gamma_{5} = -0.8$, $\gamma_{6}= \ldots=\gamma_{10} = 0.8$ and $\gamma_{11}=\ldots=\gamma_{d} = 0$;\vspace{-.1in}
    \item[(Setting 4)] $\gamma_{1}=\ldots=\gamma_{5} = 0.07$, $\gamma_{6}= \ldots=\gamma_{10} = 0.14$ and $\gamma_{11}=\ldots=\gamma_{d} = 0$.
\end{enumerate}Thus, the null hypothesis holds in the first setting, and the alternative holds in the three other settings. Calculation of parameter estimates and estimated influence functions required for inference was implemented using code adapted from the \texttt{ltmle} package in \texttt{R} \citep{lendle_ltmle_2017}.  

In the null setting (Setting 1), we find that the type one error of all norm based tests is somewhat above the nominal level, and larger for small sample sizes and high dimensions. The Bonferroni test is slightly conservative and the Cauchy test is slightly anti-conservative in all settings. In Setting 2, all tests have similar power with the adaptive tests slightly outperforming others at lower sample sizes and the Bonferroni test under-performing in all settings.  The differences in performance are larger  when dimension is higher.  In Setting 3, the $\ell_\infty$ based test almost always outperforms all other tests, with the largest differences in the sample size 500 setting.  In Setting 4, the $\ell_\infty$ and Bonferroni test under-perform all other tests, of which all have nearly identical power.

\begin{figure}
	\centering
	\includegraphics[width = \linewidth]{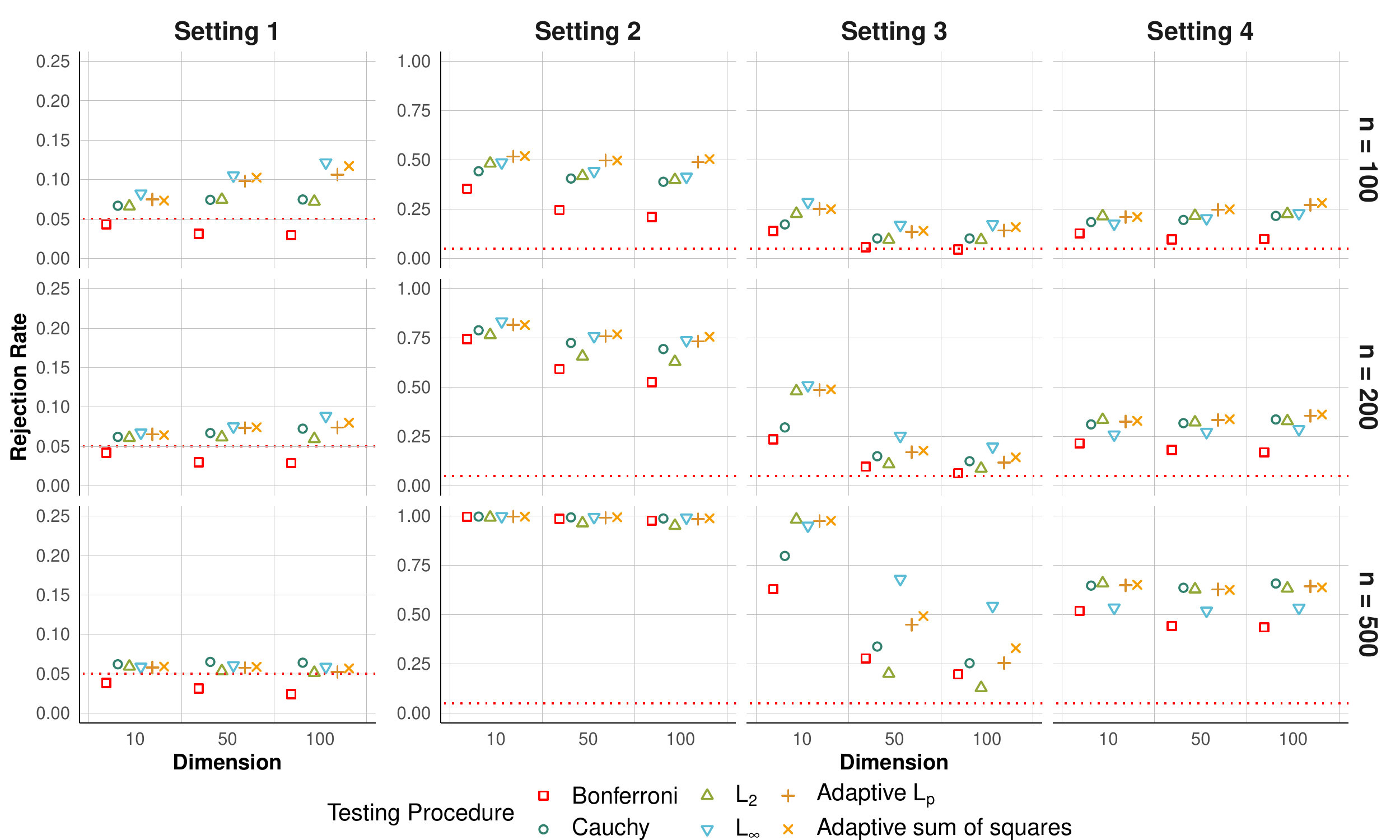}
	\linespread{1}
	\caption{Empirical rejection rates for the Bonferroni test, the Cauchy test, the non-adaptive $\ell_p$ and $\ell_\infty$ tests, and the adaptive $\ell_p$ and $\ell_\infty$ tests in Example 3 across Settings 1--4, different sample sizes and covariate vector dimensions. }
	\label{fig:examp_3}
\end{figure}

\section{Assessing correlates of risk of HIV infection in HVTN 505}
\label{sec:data_app}

Between 2008 and 2013, a cohort of 2,504 circumcised men and transgender persons who have sex with men were recruited in the United States to participate in HVTN 505, a phase IIB preventative efficacy trial of a DNA and recombinant adenovirus serotype 5 HIV vaccine  \citep{neidich_antibody_2019}.  While the vaccine under study was not found to be efficacious in preventing HIV infection, secondary analyses were conducted to study the association between the immune response to vaccine and risk of infection. This response was measured using a large number of biomarkers, including levels of various antibodies, T cells and Fc-gamma receptors. These analyses indicated the possibility of a qualitative interaction, whereby the vaccine may lower or raise the rate of HIV-1 acquisition for different subgroups, depending on the immune response  \citep{fong_modification_2018, gilbert_post-randomization_2020}. Estimates of how well each biomarker group can predict future HIV-1 infection are suggestive of which groups protect against HIV-1.

In our analysis, we consider the same groupings of biomarkers as in \cite{neidich_antibody_2019}.  For each set of biomarkers, we test the null hypothesis that no biomarker is associated with risk of infection using four tests derived from our framework.  Two of these tests are adaptive (selecting across $\ell_p$ and sum-of-squares norms, respectively), whereas the other two are non-adaptive (based on the $\ell_2$ and $\ell_\infty$ norms, respectively). The association parameter we focus on is the biomarker-specific regression coefficient from a weighted univariable working logistic regression model. Weighting accounts for the informative biomarker missingness induced by the two-phase study design. Additional details on the HVTN 505 trial and our analysis strategy are provided in the Supplement.

The results of these tests are summarised in Table \ref{tab:pval_tab_new}. Each column (except the first) corresponds to a test type and each row to a group of biomarkers considered by \citet{neidich_antibody_2019}.  With the exception of the Fx Ab and IgG$+$IgA groups, each test of association between a biomarker group and risk of infection has a $p$-value less than $0.01$ for all considered tests. For the IgG$+$IgA group, the tests yield $p$-values that are all similar, though the $\ell_2$ tests gives a slightly smaller $p$-value. The tests for the functional antibody (Fx Ab) biomarker group give similar $p$-values to one another except for the $\ell_\infty$ test, which yields a $p$-value roughly twice as large as the others. Thus, in all but one setting, the choice of testing procedure has little impact on results. For the test of the Fx Ab biomarker group, the adaptive tests provide similar $p$-values, whereas $p$-values for the non-adaptive tests differ more.

\begin{table}[]
    \centering
    \begin{tabular}{lx{2.1cm}x{2.1cm}x{2.1cm}x{2.1cm}}
\toprule
Biomarker Group & $\ell_2$ & $\ell_\infty$ & adaptive $\ell_p$ & adaptive ssq\\
\midrule
\cellcolor[HTML]{EEEEEE}{IgG + IgA} & \cellcolor[HTML]{EEEEEE}{0.127} & \cellcolor[HTML]{EEEEEE}{0.149} & \cellcolor[HTML]{EEEEEE}{0.147} & \cellcolor[HTML]{EEEEEE}{0.153}\\
IgG3 (Immuno Globulin G3 Group) & 0.000 & 0.003 & 0.000 & 0.000\\
\cellcolor[HTML]{EEEEEE}{T Cells} & \cellcolor[HTML]{EEEEEE}{0.000} & \cellcolor[HTML]{EEEEEE}{0.000} & \cellcolor[HTML]{EEEEEE}{0.000} & \cellcolor[HTML]{EEEEEE}{0.000}\\
Fx Ab & 0.062 & 0.116 & 0.052 & 0.049\\
\cellcolor[HTML]{EEEEEE}{IgG + IgA + IgG3} & \cellcolor[HTML]{EEEEEE}{0.002} & \cellcolor[HTML]{EEEEEE}{0.006} & \cellcolor[HTML]{EEEEEE}{0.002} & \cellcolor[HTML]{EEEEEE}{0.002}\\
IgG + IgA + T Cells & 0.003 & 0.000 & 0.000 & 0.001\\
\cellcolor[HTML]{EEEEEE}{IgG + IgA + IgG3 + T Cells} & \cellcolor[HTML]{EEEEEE}{0.000} & \cellcolor[HTML]{EEEEEE}{0.000} & \cellcolor[HTML]{EEEEEE}{0.000} & \cellcolor[HTML]{EEEEEE}{0.000}\\
IgG + IgA + IgG3 + Fx Ab & 0.004 & 0.004 & 0.002 & 0.002\\
\cellcolor[HTML]{EEEEEE}{T Cells + Fx Ab} & \cellcolor[HTML]{EEEEEE}{0.000} & \cellcolor[HTML]{EEEEEE}{0.000} & \cellcolor[HTML]{EEEEEE}{0.000} & \cellcolor[HTML]{EEEEEE}{0.000}\\
All markers & 0.000 & 0.001 & 0.000 & 0.000\\
\bottomrule
\end{tabular}    
    \caption{$p$-values for each combination of biomarker group and test type. This analysis is based on data from the HVTN 505 clinical trial, and the null hypothesis tested is that the biomarkers from the given group are not associated with risk of HIV infection.}
    \label{tab:pval_tab_new}
\end{table}

In Figure \ref{fig:rde_new}, we focus on the testing results for the Fx Ab group. The gray histogram in each panel shows an approximation of the estimated null limiting distribution of $Z_n$ for each considered test. The dashed red and solid black vertical lines intersect the $x$-axis at $Z^*_n$ and the $5^{th}$ percentile of the estimated limiting distribution, respectively. Both adaptive tests have distributions that are centered and more concentrated around a smaller value. Because the adaptive tests select the pointwise minimum among all norms considered, this phenomenon is expected. Figure \ref{fig:all_de_figs} of the Supplement shows this summary for every biomarker group from Table \ref{tab:pval_tab_new}.

\begin{figure}
	\centering
	\includegraphics[width = \linewidth]{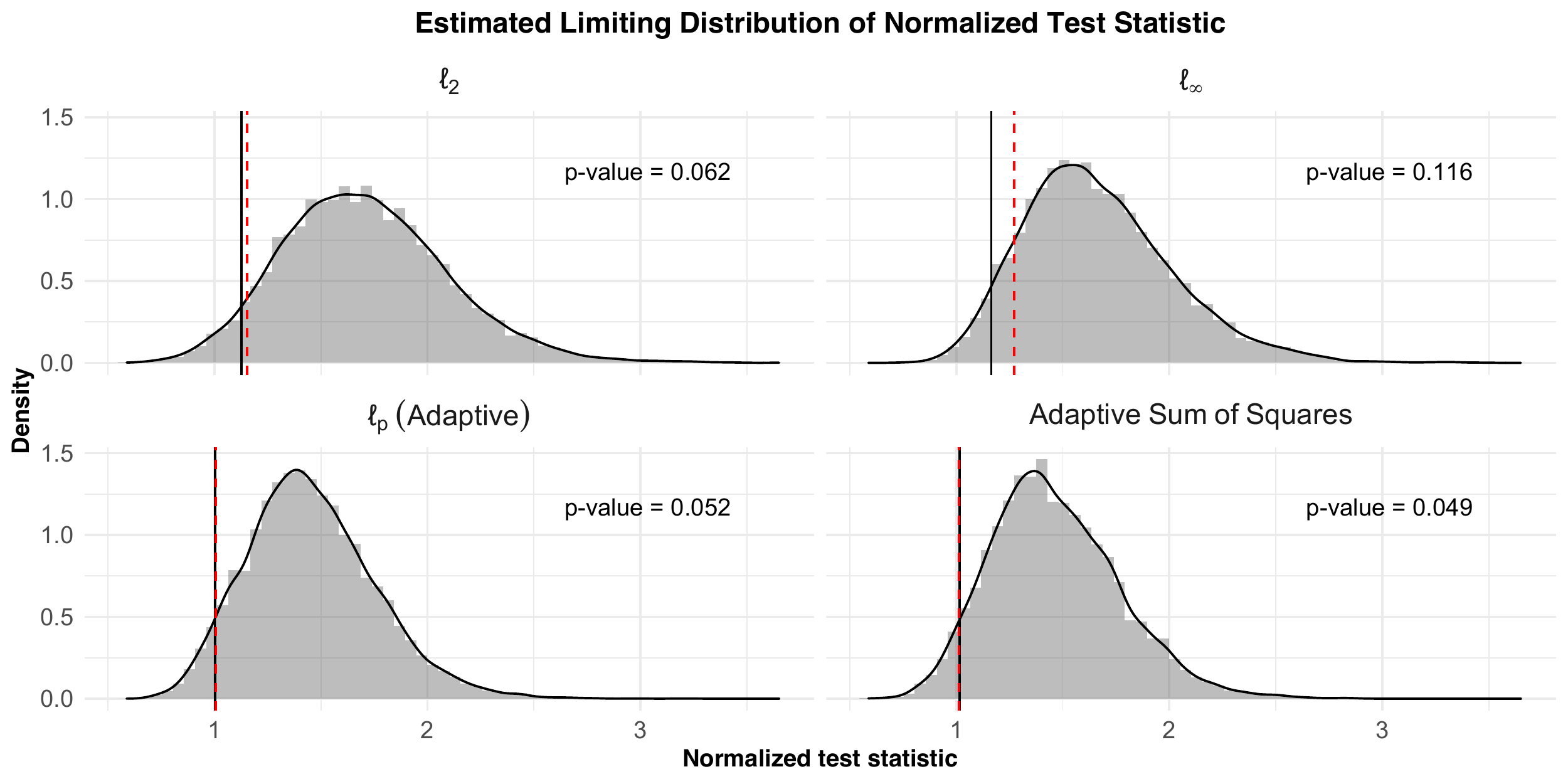}
	\linespread{1}
	\caption{Estimated limiting distributions of the multiplicative factor measure for both non-adaptive ($\ell_2$ and maximum absolute deviation) and adaptive (adaptive $\ell_p$ and adaptive sum-of-squares) testing procedures.  The black vertical line in each plot represents the $0.05$ quantile of the limiting distribution, and the dashed red vertical line represents the value of the test statistic.  This analysis is based on data from the HVTN 505 clinical trial, and the null hypothesis tested is that the biomarkers from the Fx Ab group are not associated with risk of HIV infection.}
	\label{fig:rde_new}
\end{figure}

\section{Concluding remarks}
\label{sec:discuss}
We have described a general framework for constructing tests of a multivariate point null hypothesis in settings in which an asymptotically linear estimator of the underlying target parameter is available.  Tests created using this framework leverage knowledge of the parameter estimator and its estimated sampling distribution to adaptively build a test statistic that provides good power under alternatives suggested by the data. Tests constructed using our framework have desirable asymptotic guarantees under the null, fixed alternatives, and local alternatives. We studied the performance of tests constructed using our framework in simulation studies and find these tests have comparable performance to tailor-made methods in settings in which specialized methods currently exist and have favorable properties in settings in which they do not.

The framework we described is quite general, allowing users to specify the parameter of interest and to utilize as much or as little information of the data-generating mechanism as is known. However, it does require an estimator of the covariance matrix of the parameter estimator. For most common parameters, such (non-parametric) estimators already exist and in novel settings constructing these estimators can be facilitated using influence functions. Such analytic derivations could pose a challenge for the implementation of this method in novel settings, though work has been done to allow for such computations to be carried out numerically, which could remove this hurdle \citep{carone_toward_2019}. 

Finally, while we have focused on point null hypotheses in this paper, our proposed framework can also be used to test certain composite null hypotheses. For example, the composite null hypothesis $H_0:\psi_{10}=\psi_{20}=\ldots=\psi_{d0}$ can be equivalently stated as $H_0:\psi^*_{10}=\psi^*_{20}=\ldots=\psi^*_{(d-1)0}=0$, where we define $\Psi_j^*(P):=\Psi_j(P)-\Psi_d(P)$ and write $\psi^*_{j0}:=\Psi_j^*(P_0)$. Indeed, whenever the composite null hypothesis of interest can be restated as a point null hypothesis (of lower dimension) based on a different parametrization, the methods we have proposed can be used directly.

\vspace{0.1in}
\singlespacing
{\small
\section*{Acknowledgements}
The authors would like to thank Brian Williamson for his generous help providing data and guidance on the data analysis. This work was supported by NIH grants DP2-LM013340, R01-HL137808, R01-AI029168 and UM1-AI068635. The opinions expressed in this article are those of the authors and do not necessarily represent the official views of the NIH.

\vspace{0.1in}

\bibliographystyle{chicago}
\bibliography{refs}
}

\doublespacing

\newpage

\section*{Supplement}

\subsection*{Technical lemmas}

We first state and prove technical lemmas that will be used in the proof of Theorems 1, 2 and 3. The first lemma below indicates when regularity conditions on the individual local measures of test inefficiency in $\mathscr{F}_0$ imply corresponding conditions for the adaptive local measure of test inefficiency upon which our test is derived. This lemma serves as a fundamental building block in the proof of Theorems 1 and 2.

Throughout the Supplement, we define $\Gamma^{*}:(x, \Sigma)\mapsto\min_{\varphi\in\mathscr{F}_0}\Gamma(x, \Sigma,\varphi)$ and $\Gamma_0^*:x\mapsto \Gamma^*(x,\Sigma_0)$. Also, for any event $\mathscr{E}$ involving only the random variable $U_0$, we denote by $\text{pr}(\mathscr{E})$ the evaluation of the probability of $\mathscr{E}$ under $U_0\sim Q_0$. Finally, we refer to the following conditions: \begin{enumerate}[label= {C\arabic{enumi}}*), ref = C\arabic{enumi}*]
	\item $(u,\Sigma)\mapsto \Gamma^*(u,\Sigma)$ is continuous and non-negative on $\mathbb{R}^d \times B^*_0$ for some neighborhood $B^*_0\subset \mathbb{V}_d$  of $\Sigma_0$; \label{itm:cond_cont_non_neg_star}
	\item $\int I\{\Gamma^*_0(u)=t\}\,dQ_0(u)=0$ for every $t\geq 0$;\label{itm:cond_dom_by_lebesgue_star}
	\item for at least one $\varphi\in\mathscr{F}_0$, $\Gamma^*(x_s, \Sigma)\rightarrow 0$ uniformly over $\Sigma \in B^*_1$ for some neighborhood $B^*_1\subset \mathbb{V}_d$ of $\Sigma_0$ for every sequence $x_1,x_2,\ldots$ of elements of $\mathbb{R}^d$ such that $\varphi(x_s)\rightarrow \infty$; \label{itm:cond_conv_prb_zr_star} 
	\item $u\mapsto \Gamma^*_0(u)$ is quasi-concave, in the sense that $\{u : \Gamma^*_0(u) \geq a\}$ is convex for every $a\geq 0$; \label{itm:cond_unimod_star}
	\item $u\mapsto \Gamma^*_0(u)$ is centrally symmetric, in the sense that $\Gamma^*_0(u)=\Gamma^*_0(-u)$ for every $u\in\mathbb{R}^d$. \label{itm:cond_cent_sym_star}
\end{enumerate}

\begin{lemma}\label{lem:Gammastar}
If any combination of \ref{itm:cond_cont_non_neg}, \ref{itm:cond_dom_by_lebesgue}, \ref{itm:cond_unimod} and \ref{itm:cond_cent_sym} hold for every element of $\mathscr{F}_0$, then the respective combination of \ref{itm:cond_cont_non_neg_star}, \ref{itm:cond_dom_by_lebesgue_star}, \ref{itm:cond_unimod_star} and \ref{itm:cond_cent_sym_star} hold as well. Additionally, if \ref{itm:cond_conv_prb_zr} holds for at least one element of $\mathscr{F}_0$, then \ref{itm:cond_conv_prb_zr_star} holds as well.
\end{lemma}

\begin{proof}[Proof of Lemma~\ref{lem:Gammastar}]
Suppose that \ref{itm:cond_cont_non_neg} holds for each $\varphi\in\mathscr{F}_0$.  Denote by $B_0(\varphi)\subset \mathbb{V}_d$ the neighborhood of $\Sigma_0$ over which \ref{itm:cond_cont_non_neg} holds for $\varphi\in\mathscr{F}_0$. Because the minimum function is continuous and the composition of continuous functions is also continuous, it follows that $\Gamma^*$ is continuous on $\mathbb{R}^d\times B_0^*$, where $B_0^*:=\cap_{\varphi\in\mathscr{F}_0}B_0(\varphi)\subset \mathbb{V}_d$ is itself a neighborhood of $\Sigma_0$. Additionally, the minimum of non-negative values is necessarily non-negative. Thus,  \ref{itm:cond_cont_non_neg_star} holds. Suppose that \ref{itm:cond_dom_by_lebesgue} holds for each $\varphi\in\mathscr{F}_0$.  Because $\text{pr}\{\Gamma_0(U_0,\varphi) = c\} = 0$ for each $c\in\mathbb{R}$ and $\varphi \in \mathscr{F}_0$, it follows that 
\begin{align*}
	\textstyle\text{pr}\left\{\Gamma_0^*(U_0) = c\right\}\ \leq\ \text{pr}\left\{ \bigcup_{\varphi\in\mathscr{F}_0} \left\{\Gamma_0(U_0, \varphi) = c \right\} \right\}\ \leq\ \sum_{\varphi\in\mathscr{F}_0} \text{pr}\left\{\Gamma_0(U_0, \varphi) = c\right\}\ =\ 0\ ,
\end{align*}establishing \ref{itm:cond_dom_by_lebesgue_star}.  Suppose that \ref{itm:cond_conv_prb_zr} holds for $\varphi_0\in\mathscr{F}_0$, and denote by $B^*_1\subset\mathbb{V}_d$ the neighborhood of $\Sigma_0$ over which \ref{itm:cond_conv_prb_zr} holds.  By definition, we have that $0\leq \Gamma^*(x_s,\Sigma)\leq \Gamma(x_s,\Sigma,\varphi_0)$ for each $\Sigma$, and so, $0\leq \sup_{\Sigma\in B_1^*}\Gamma^*(x_s,\Sigma)\leq\sup_{\Sigma\in B_1^*}\Gamma(x_s,\Sigma,\varphi_0)$. This establishes \ref{itm:cond_conv_prb_zr_star} since $\sup_{\Sigma\in B_1^*}\Gamma(x_s,\Sigma,\varphi_0)\rightarrow 0$ by assumption.
Suppose that \ref{itm:cond_unimod} holds for each $\varphi\in\mathscr{F}_0$. Defining $M_\varphi(a):=\{u\in\mathbb{R}^d:\Gamma_0(u,\varphi)\geq a\}$ and $M^*(a):=\{u\in\mathbb{R}^d:\Gamma^*_0(u)\geq a\}$, we note that $M^*(a)=\cap_{\varphi\in\mathscr{F}_0}M_\varphi(a)$. By assumption, $M_\varphi(a)$ is convex for each $a\geq 0$ and $\varphi\in\mathscr{F}_0$. Since the  intersection of a finite number of convex sets is convex, $M^*(a)$ is convex for each $a\geq 0$, thus proving \ref{itm:cond_unimod_star}. Suppose that \ref{itm:cond_cent_sym} holds for each $\varphi\in\mathscr{F}_0$. Because $u\mapsto \Gamma_0(u,\varphi)$ is centrally symmetric for each $\varphi\in\mathscr{F}_0$, we have that $\Gamma^*_0(-u)=\min_{\varphi\in\mathscr{F}_0}\Gamma_0(-u,\varphi)=\min_{\varphi\in\mathscr{F}_0}\Gamma_0(u,\varphi)=\Gamma^*_0(u)$ for each $u\in\mathbb{R}^d$, and so, \ref{itm:cond_cent_sym_star} holds.
\end{proof}

The following lemmas establish technical properties for certain sets, functions and probability statements considered in the proof of Theorems 1 and 2.

\begin{lemma}
	\label{lemma:mvtnorm_quasi-concavity}
	The density function $f$ of the $d$-variate normal distribution with mean zero  is quasi-concave, that is, the set $\{x\in\mathbb{R}^d:f(x)\geq \kappa\}$ is convex for each $\kappa\in\mathbb{R}$.
\end{lemma}
\begin{proof}[Proof of Lemma \ref{lemma:mvtnorm_quasi-concavity}]
    By \cite{tong_multivariate_2012}, the $d$-variate normal probability density function is log-concave. All log-concave functions are quasi-concave in view of  Section 3.5.1 of \cite{boyd_convex_2004}.
\end{proof}

\begin{lemma}
	\label{lemma:set_convexity}
	Let $C$ be a convex subset of $\mathbb{R}^d$, and define $C_{ \mu} := \{ c +  \mu :  c \in C\}$ for $ \mu\in \mathbb{R}^d$. For any $\mu_1,\mu_2\in\mathbb{R}^d$ and $t\in[0,1]$, the set $t C_{ \mu_1} + (1 - t) C_{ \mu_2} = \{t c_{1} + (1 - t) c_{2}:  c_{1} \in C_{ \mu_1},  c_{2} \in C_{ \mu_2}\}$ is equal to $C_{t \mu_1 + (1 - t)  \mu_2}$.
\end{lemma}
\begin{proof}[Proof of Lemma \ref{lemma:set_convexity}]
	Let $x \in t C_{ \mu_1} + (1 - t) C_{ \mu_2}$, so that there exist $ c_1, c_2\in C$ such that $x=( c_1 +  \mu_1) t + ( c_2 +  \mu_2)(1 - t)$. Since we can rewrite $x= c_1 t +  c_2 (1 - t) +  \mu_1 t +  \mu_2 (1 - t)$ with $ c_1 t +  c_2 (1 - t)\in C$ by the convexity of $C$, we have that $x \in C_{ \mu_1 t +  \mu_2 (1 - t)}$. Hence, we find that $t C_{ \mu_1} + (1 - t) C_{ \mu_2}\subseteq C_{ \mu_1 t +  \mu_2 (1 - t)}$. To show the reverse inclusion, let $ y \in C_{ \mu_1 t +  \mu_2 (1 - t)}$, so that there exists $ c\in C$ such that
	$y= c +  \mu_1 t +  \mu_2 (1 - t) = ( c +  \mu_1) t + ( c +  \mu_2) (1 - t)$.
	This implies that $ y \in t C_{ \mu_1} + (1 - t) C_{ \mu_2}$. Hence, we also find that $C_{ \mu_1 t +  \mu_2 (1 - t)}\subseteq t C_{ \mu_1} + (1 - t) C_{ \mu_2}$.
\end{proof}

\begin{lemma}
\label{lemma:aditional_anderson_cond}
Suppose that $B$ is a closed, bounded and centrally symmetric subset of $\mathbb{R}^d$, and let $f$ denote the density function of the $d$-dimensional normal distribution with mean zero and positive definite covariance matrix. For any non-zero $h\in\mathbb{R}^d$, the function $g_{h}:\beta\mapsto \int_B f(t-\beta h)\,dt$ is strictly decreasing.
\end{lemma}
\begin{proof}[Proof of Lemma \ref{lemma:aditional_anderson_cond}]

A minimizer $x_0 \in \text{argmax}\left\{ f(x) : x \in \convset\right\}$ exists because $\convset$ is closed and bounded and $ f$ is continuous. We also have that $-x_0 \in \text{argmax}\left\{ f(x) : x \in \convset\right\}$ because $\convset$ and $ f$ are both centrally symmetric. Let $\Sigma$ be the covariance matrix indexing $ f$, and define $x_0^*:= {\rm sign}( h^{\top}\Sigma^{-1}x_0)\cdot x_0\in B$. In particular, we note that $x_0^*+h\in \{x+h:x\in B\}$. We also note that
\begin{align*}
    2[\log f(x^*_0) - \log f(x^*_0 + h)]\ &=\  (x_0^* + h)^{\top}\Sigma^{-1}(x_0^* + h) -(x_0^*)^{\top}\Sigma^{-1}x_0^* = 2 h^{\top}\Sigma^{-1}x_0^* + h^{\top}\Sigma^{-1} h \\
    &=\ 2|h^{\top}\Sigma^{-1}x_0| + h^{\top}\Sigma^{-1} h\ >\ 0\ ,
\end{align*}
and so, $ f(x^*_0) >  f(x^*_0 + h)$ for $h\in\mathbb{R}^d\neq 0$. Set $u:=[ f(x^*_0)+ f(x^*_0 + h)]/2$ and note that $ f(x_0^*)>u> f(x_0^*+h)$. This implies that $x_0^*+h$ is an element of $\{x+h:x\in B, f(x)>u\}$ but not of $\{x+h:x\in B\}\cap\{x: f(x)>u\}$. Thus, in view of Corollary 1 of \cite{anderson_integral_1955}, $g_h$ is strictly decreasing.
\end{proof}
\begin{lemma}
\label{lemma:zero_measure_ray_monotonic}
Let constants $b,c\in\mathbb{R}$ and strictly ray monotone function $g:\mathbb{R}^d\rightarrow \mathbb{R}$ be given. If $Z$ is a non-degenerate $d$-variate normal random variable with distribution $P$ and density function $ f$,  it holds that $P\left\{g(Z) = c\right\}=\int I\{g(z) = c\}  f(z)\,dz=0$.
\end{lemma}

\begin{proof}[Proof of Lemma \ref{lemma:zero_measure_ray_monotonic}]
We focus on $d\ge 3$ as the case $d\in\{1,2\}$ is straightforward. To evaluate the integral $\int I\{g(z) = c\}  f(z)\,dz$, we reparametrize $\mathbb{R}^d$ into $\{(r,\gamma_1,\gamma_2,\ldots,\gamma_{d-1}):r\in[0,\infty),\gamma_1,\gamma_2,\ldots,\gamma_{d-2}\in[0,\pi),\gamma_{d-1}\in[0,2\pi)\}$ as in  \citet{blumenson_derivation_1960}, setting  $z=rt(\gamma)$ with  $t_1(\gamma):=\cos(\gamma_1)$, $t_j(\gamma):=\sin(\gamma_1)\sin(\gamma_2)\ldots\sin(\gamma_{j-1})\cos(\gamma_{j})$ for $j=2,3,\ldots,d-1$, and $t_d(\gamma):=\sin(\gamma_1)\sin(\gamma_2)\ldots\sin(\gamma_{d-1})$. Here, we write $\gamma:=(\gamma_1,\gamma_2,\ldots,\gamma_{d-1})$ and $t(\gamma):=(t_1(\gamma),t_2(\gamma),\ldots,t_d(\gamma))$. We also define $|J_d| := r^{d-1} \text{sin}(\gamma_1)^{d - 2}\text{sin}(\gamma_2)^{d - 3} \dots \text{sin}(\gamma_{d - 2})$, and  note that $a_{ \gamma}:= |J_d|r^{1-d}$ depends on $\gamma$ but not $r$.  This change of variable allows us to write
\begin{align*}
\int I\{g(z) = c\}  f(z)\,dz\ &=\ \int_{0}^{2\pi}\int_{0}^{\pi} \dots \int_{0}^{\pi}  \int_{0}^{\infty}  I\left\{g\left(r t( \gamma)
\right) = c\right\}  f\left(r t( \gamma) \right) |J_d|\,dr\, d\gamma_1 \dots d\gamma_{d - 1} \\
&= \int_{0}^{2\pi}\int_{0}^{\pi} \dots \int_{0}^{\pi}  a_{ \gamma}\left[\int_{0}^{\infty} r^{d-1} I\left\{g\left(r t( \gamma)
\right) = c\right\}  f\left(r t( \gamma) \right) \,dr\right] d\gamma_1 \dots d\gamma_{d - 1}\ ,
\end{align*}
where we note that in the innermost integral $t( \gamma)$ is fixed when integrating over $r$.  Thus, the latter integral has the form $
\int_0^\infty r^{d - 1}I\{g(r v) = c\}  f(r v)\, dr$ for some $v\in\mathbb{R}^d$. 
Because $g$ is strictly ray monotone, the function $g_{ v}:r\mapsto g(rv)$ is strictly monotone.  Thus, the indicator function $I\{g(r v) = c\}$ can only equal one for a single value of $r$, and so, the innermost integral and thus the probability of interest equal zero.
\end{proof}

\subsection*{Proof of Theorems \ref{thm:t1ec} and \ref{thm:unbiased_locl_alt}}

We now prove Theorems 1 and 2. Below, we denote convergence in distribution by $\rightsquigarrow$. We refer to $\bar{U}_n$ as a random draw from the normal distribution with mean zero and covariance matrix $\Sigma_n$ independent of $X_1,X_2,\ldots,X_n$ given $\Sigma_n$. We define $\bar{A}^*_n:=\Gamma^*(\bar{U}_n,\Sigma_n)$, $A^*_n:=\Gamma^*(U_n,\Sigma_n)$ and $A^*_0:=\Gamma^*(U_0,\Sigma_0)$, and denote by $q_n$ and $q_0$ the $\alpha$-quantile of $\bar{A}^*_n$ and $A^*_0$, respectively.

\begin{proof}[Proof of Theorem~\ref{thm:t1ec}]
In view of Lemma \ref{lem:Gammastar}, if conditions \ref{itm:cond_cont_non_neg_star}--\ref{itm:cond_dom_by_lebesgue_star} hold for each $\varphi\in\mathscr{F}_0$, then  \ref{itm:cond_cont_non_neg_star}--\ref{itm:cond_dom_by_lebesgue_star} hold, and if in addition \ref{itm:cond_conv_prb_zr} holds for some $\varphi\in\mathscr{F}_0$, then  \ref{itm:cond_conv_prb_zr_star}  holds as well. Since $\Sigma_n$ converges in probability to $\Sigma_0$, we have that $(\bar{U}_n,\Sigma_n)$ converges in distribution to $(U_0,\Sigma_0)$.  In view of \ref{itm:cond_cont_non_neg_star}, this implies that $\bar{A}^*_n\rightsquigarrow A^*_0$ by the continuous mapping theorem.  By Lemma 21.2 of \cite{van_der_vaart_asymptotic_2000}, we have that $q_n\rightarrow q_0$ in probability since the distribution function of  $A^*_0$ is continuous at $q_0$ by \ref{itm:cond_dom_by_lebesgue_star}. Suppose that $P_0 \in \mathcal{M}_0$. The asymptotic linearity of $ \psi_n$ and consistency of $\Sigma_n$ imply that $(U_n,\Sigma_n)\rightsquigarrow(U_0,\Sigma_0)$, and so, $A^*_n\rightsquigarrow A^*_0$ by the continuous mapping theorem in view of \ref{itm:cond_cont_non_neg_star}. By \ref{itm:cond_dom_by_lebesgue_star}, this implies that the type I error $P_0(A^*_n \leq q_n)$ tends to $P_0(A^*_0 \leq q_0) = \alpha$.

Now, suppose instead that $P_0 \not\in \mathcal{M}_0$. For any $\varepsilon>0$, the rejection rate $P_0(A^*_n \leq q_n)=1 - P_0(A^*_n > q_n)$ is bounded below by
$ 1 - P_0(A^*_n \geq q_0 - \varepsilon) -  P_0(|q_n - q_0|> \varepsilon)$.  The term $P_0(|q_n - q_0|> \varepsilon)$ tends to zero in  probability by the consistency of $q_n$ for $q_0$. Since $q_0 > 0$ by 
 \ref{itm:cond_dom_by_lebesgue_star}, we can choose $\varepsilon\in (0, q_0)$ above, and for any such choice, we have that
\[
    \liminf_{n\rightarrow\infty}P_0(A^*_n\leq q_n)\ge 1-\limsup_{n\rightarrow\infty}P_0(A^*_n \geq q_0 - \varepsilon)\ .
\]It then directly follows that the power $P_0(A^*_n\leq q_n)$ of the test tends to one provided $A^*_n$ tends to zero in probability. We thus show that the latter statement holds. First note that a sequence of random variables converges in probability if and only if each subsequence of this sequence contains a further subsequence that converges almost surely to the same limit \citep{shorack_probability_2017}. 
Let $\varphi_0\in\mathscr{F}_0$ be such that $\Gamma^*(x_s, \Sigma)\rightarrow 0$ uniformly over $\Sigma \in B^*$ for some neighborhood $B^*$ of $\Sigma_0$ for every sequence $x_1,x_2,\ldots$ of elements of $\mathbb{R}^d$ such that $\varphi_0(x_s)\rightarrow \infty$; the existence of $\varphi_0$ is guaranteed by \ref{itm:cond_conv_prb_zr_star}.
By the reverse triangle inequality,  we have that $\varphi_0(n^{1/2}\psi_n)\ge \varphi_0(n^{1/2} \psi_0)- \varphi_0(n^{1/2} (\psi_n- \psi_0)) = n^{1/2}\varphi_0( \psi_0)-O_P(1)$ in view of the fact that $n^{1/2}(\psi_n-\psi_0)$ is bounded in probability. As $n^{1/2}\varphi_0( \psi_0)\rightarrow+\infty$, this shows that $V_n:=1/\varphi_0(n^{1/2} \psi_n)$ tends to zero in probability.  Let $V_{n_1},V_{n_2},\ldots$  with $1\leq n_1<n_2<\ldots$ be an arbitrary subsequence of $V_1,V_2,\ldots$, and  note that $V_{n_k}$ tends to zero in probability as $k\rightarrow\infty$. There must then exist a further subsequence $V_{n_{k_1}},V_{n_{k_2}},\ldots$ with $1\leq k_1<k_2<\ldots$ that converges to zero almost surely, and so, defining $U'_j:=U_{n_{k_j}}=n_{k_j}^{1/2}\psi_{n_{k_j}}$, we have that $\varphi_0(U'_j)$ diverges almost surely as $j\rightarrow \infty$. Thus, it follows that $\Gamma^*(U'_j, \Sigma)$ converges to zero almost surely uniformly over $\Sigma \in B^*$. Since we have argued that every subsequence $\Gamma^*(U_{n_k},\Sigma)$ has a further subsequence that converges to zero almost surely, we have shown that
$\Gamma^*(U_n, \Sigma)$ tends to zero in probability uniformly over $\Sigma \in B^*$. For each $\varepsilon > 0$, we then have that
\begin{align*}
    P_0(A^*_n > \varepsilon)\ &=\  P_0(A^*_n>\varepsilon,\Sigma_n\in B^*)  +  P_0(A^*_n>\varepsilon,\Sigma_n\notin B^*)\\
    &\leq\ \sup_{\Sigma \in B^*}P_0\{\Gamma^*(U_n, \Sigma) > \varepsilon\}  +  P_0( \Sigma_n \not \in B^*)\ ,
\end{align*}which implies the claim since the first and second summands tend to zero in view of \ref{itm:cond_conv_prb_zr_star} and the consistency of $\Sigma_n$, respectively.
\end{proof}

\begin{proof}[Proof of Theorem \ref{thm:unbiased_locl_alt}]
In view of Lemma \ref{lem:Gammastar}, the fact that conditions \ref{itm:cond_cont_non_neg}, \ref{itm:cond_dom_by_lebesgue}, \ref{itm:cond_unimod} and \ref{itm:cond_cent_sym} hold for each $\varphi\in\mathscr{F}_0$ and that condition  \ref{itm:cond_conv_prb_zr} holds for some $\varphi\in\mathscr{F}_0$ implies that \ref{itm:cond_cont_non_neg_star}--\ref{itm:cond_cent_sym_star} hold.  Since $\Sigma_n$ is consistent and $\psi_n$ is regular and asymptotically linear, under any sequence $P_n^{(0)}$ of local alternatives, it holds that $U_n \rightsquigarrow U_0 +  h$ for some $h\in\mathbb{R}^d$, and in view of \ref{itm:cond_cont_non_neg_star}, $A^*_n=\Gamma^*(U_n,\Sigma_n) \rightsquigarrow \Gamma_0^*(U_0 +  h)$. Also, in view of \ref{itm:cond_cont_non_neg_star} and \ref{itm:cond_dom_by_lebesgue_star}, it follows that $\bar{A}_n^* \rightsquigarrow A^*_0$ and the distribution function of $A^*_0$ is continuous. Therefore, $q_n$ tends to $q_0$ in probability under this sequence of local alternatives. Lastly, $q_0$ is a continuity point of $\Gamma_0^*(U_0 + h)$ using a change of variables argument and \ref{itm:cond_dom_by_lebesgue_star}. Thus, it follows that  $P_n^{(0)}(A^*_n \leq q_n) \to \textrm{pr}\left\{\Gamma_0^*(U_0+h) \leq q_0\right\}$. We define the function $g_{ h} : \mathbb{R} \to \mathbb{R}$ pointwise as
\begin{align}
	g_{ h}(\beta) := \pr\left\{\Gamma^*_0(U_0 + \beta  h) \geq q_0\right\}  &= \int_{B_0} f_{0}(x - \beta  h)\,dx\ , \label{eqn:decr_part}
\end{align}where we define $B_0:=\{x\in\mathbb{R}^d:  \Gamma^*_0(x) \geq q_0\}$ and denote by $ f_{0}$ the density of the multivariate normal distribution with mean zero and covariance $\Sigma_0$. By Lemma \ref{lemma:mvtnorm_quasi-concavity}, the probability density function of the multivariate normal distribution with mean zero is centrally symmetric and quasi-concave, and so, $g_{ h}$ is non-increasing in view of Theorem 1 of \citet{anderson_integral_1955}.  Corollary 1 of \citet{anderson_integral_1955} states that $g_h$ is in fact strictly decreasing provided $\{\omega +  h : \omega \in B_0,  f_{0}(\omega) > u\} \neq \{\omega +  h : \omega \in B_0\} \cap \{\omega+h :  f_0(\omega) > u\}$.  Lemma \ref{lemma:aditional_anderson_cond} indicates that this condition is satisfied if $\convset_0$ is closed, bounded and centrally symmetric. By \ref{itm:cond_cent_sym_star}, $\Gamma^*_0$ is centrally symmetric, and so,  $\convset_0$ is also centrally symmetric.  Also, since $u\mapsto \Gamma^*_0(u)$ is continuous, it is also upper semicontinuous, and therefore, $B_0$ is closed. It remains to show that $B_0$ is bounded. Let $\varphi_0\in\mathscr{F}_0$ be such that $\Gamma^*(x_s, \Sigma)\rightarrow 0$ uniformly over $\Sigma \in B^*$ for some neighborhood $B^*$ of $\Sigma_0$ for every sequence $x_1,x_2,\ldots$ of elements of $\mathbb{R}^d$ such that $\varphi_0(x_s)\rightarrow \infty$; the existence of $\varphi_0$ is guaranteed by \ref{itm:cond_conv_prb_zr_star}. Suppose that $B_0$ is not bounded, that is, for each $r=1,2,\ldots$, there exists some $ v_r \in B_0$ for which $\varphi_0( v_r) > r$. Because $\varphi_0( v_r) \rightarrow \infty$,  it follows that $\sup_{\Sigma\in B^*}\Gamma( v_r, \Sigma,\varphi_0) \to 0$, and since $0\leq \Gamma^*(v_r,\Sigma)\leq \Gamma( v_r, \Sigma,\varphi_0)$ for each $\Sigma\in B^*$, this also implies that $\sup_{\Sigma\in B^*}\Gamma^*(v_r,\Sigma)\rightarrow 0$.  However, this is a contradiction since by definition $\sup_{\Sigma\in B^*}\Gamma^*(v_r,\Sigma)\geq \Gamma_0^*( v_r) \geq q_0$ for every $r$.  Thus, no such sequence exists, and instead there exists some $r_0>0$ such that $\varphi_0(v) < r_0$  for every $ v \in B_0$. Thus, $B_0$ must be bounded. It follows finally that $g_h$ is strictly decreasing, and so, the power of the proposed test under local alternatives tends to $1 - g_h(1) > 1 - g_h(0)=\text{pr}\{\Gamma_0^*(U_0)< q_0\}=\alpha$.
\end{proof}

\subsection*{Proof of Theorem \ref{thm:performance_metric}}
\label{sec:perfm_metr_prf}
We now show that both local measures of test inefficiency discussed in this paper satisfy regularity conditions \ref{itm:cond_cont_non_neg}--\ref{itm:cond_cent_sym} irrespective of the norm $\varphi$ used. 

\begin{proof}[Proof of Theorem \ref{thm:performance_metric}]
\underline{\textbf{Part 1:}} acceptance rate measure.\\
\textbf{C1.} Non-negativity is clear. To establish the continuity of $(x,\Sigma)\mapsto \Gamma_{\text{ar}}(x,\Sigma,\varphi)$, we first show that $\zeta:(\garg, \Sigma, c)\mapsto\int I\{\varphi(t) < c\}  f_\Sigma(t- x)\,dt$
is continuous, where $ f_\Sigma$ is the density function of the $d$-dimensional normal distribution with mean zero and covariance matrix $\Sigma$. Fix $x_0\in \mathbb{R}^d$, $\Sigma_0'\in \mathbb{V}_d$ sufficiently close to $\Sigma_0$ to ensure that it is invertible, and $c_0\in\mathbb{R}$. Consider an arbitrary sequence $(x_1, \Sigma_1, c_1),(x_2,\Sigma_2,c_2),\ldots$ in $\mathbb{R}^d\times \mathbb{V}_d\times \mathbb{R}$ tending to $(x_0, \Sigma_0', c_0)$. Since the smallest eigenvalue of $\Sigma_j$ converges to that of $\Sigma_0'$,  $\Sigma_j$ is invertible for all $j$ sufficiently large. Hence, without loss of generality, we suppose that $\Sigma_j$ is invertible for all $j$. By the triangle inequality, for any $j$, we have that
\begin{align}
    &|\zeta(x_j, \Sigma_j, c_j)-\zeta(x_0, \Sigma'_0, c_0)|\ \leq\ |\zeta(x_j, \Sigma_j, c_j)-\zeta(x_0, \Sigma'_0, c_j)| + |\zeta(x_0, \Sigma'_0, c_j)-\zeta(x_0, \Sigma'_0, c_0)|\ . \label{eq:seqcont}
\end{align}We first show that $\sup_{c\in\mathbb{R}} |\pregam(\garg_j, \Sigma_j, c)-\pregam(\garg_0, \Sigma_0, c)|\rightarrow 0$, which implies that the first summand tends to zero. Let $T_1,T_2,\ldots$ be independent random vectors with $T_j$ following the $d$-dimensional normal distribution with mean $x_j$ and covariance matrix $\Sigma_j$, and let $T_0$ be an independent random vector following the $d$-dimensional normal distribution with mean $x_0$ and covariance matrix $\Sigma_0'$. Because the moment generating function of $ T_j$ converges pointwise to that of $ T_0$, we have that $ T_j\rightsquigarrow T_0$, and by the continuity of norms, it follows that $\varphi( T_j)\rightsquigarrow\varphi( T_0)$. Hence, the  distribution function $F_j$ of $\varphi( T_j)$ tends to the distribution function $F_0$ of $\varphi( T_0)$ at all continuity points of $F_0$. Because $ T_j$ is a non-degenerate normal random vector and all norms are strictly ray increasing, $F_j$ is everywhere continuous for each $j$ in view of  Lemma~\ref{lemma:zero_measure_ray_monotonic}, so that $F_j(c)\rightarrow F_0(c)$ for each $c\in\mathbb{R}$. Moreover, by Lemma~2.11 in \citet{van_der_vaart_asymptotic_2000}, this convergence is uniform, that is, $\sup_{c\in\mathbb{R}}|F_j(c)-F_0(c)|\rightarrow 0$. 
 Since the continuity of $F_j$ everywhere implies that $\zeta(\garg_j, \Sigma_j, c_0)=F_j(c_0)$ for each $j$, we find that $\sup_{c\in\mathbb{R}}|\zeta(x_j, \Sigma_j, c)-\zeta(x_0, \Sigma_0, c)|\rightarrow 0$, as claimed. That the second summand in \eqref{eq:seqcont} also tends to zero follows from the fact that $\zeta(x_0, \Sigma_0, c_j)=F_0(c_j)\rightarrow F_0(c_0)=\zeta(x_0, \Sigma_0, c_0)$ since $c_0$ is necessarily a continuity point of $F_0$. 

For each $\Sigma\in\mathbb{V}_d$, we define $c_\alpha(\Sigma):=\min\{c>0:\text{pr}\{\varphi(U)<c\}\geq 1-\alpha\}$ with $U$ a multivariate normal random vector with mean zero and covariance $\Sigma$. We wish to show that $c_\alpha$ is continuous at $\Sigma_0$. We first note that $c_\alpha(\Sigma_j)$ is the $(1-\alpha)$-quantile of $\varphi( T_j)$ in the setting in which $x_1=x_2=\ldots= 0$.  Since we have already shown that the distribution function of $\varphi( T_j)$ converges uniformly to that of $\varphi( T_0)$,  it follows from Lemma 21.2 of \citet{van_der_vaart_asymptotic_2000} that the quantile function $F_j^{-1}$ of $\varphi(T_j)$ converges to the quantile function $F_0^{-1}$ of $\varphi( T_0)$.  Thus, we have that
$
c_{\alpha}(\Sigma_j) = F^{-1}_j(1 - \alpha) \rightarrow F^{-1}_0(1 - \alpha) = c_{\alpha}(\Sigma'_0)
$, thereby establishing that $c_\alpha$ is a continuous function in a neighborhood of $\Sigma
_0$.
Since $\Gamma_{\text{ar}}(x,\Sigma,\varphi)=\zeta(x,\Sigma,c_\alpha(\Sigma))$ for each $(x,\Sigma)$, $\Gamma_{\text{ar}}$ is a composition of continuous functions, thereby implying \ref{itm:cond_cont_non_neg}.\\
\textbf{C2.} Fix $x\in\mathbb{R}^d$ and define $g_x:\beta\mapsto \int_W  f_{\Sigma_0}(t-\beta x)\,dt$ with $W:=\{t\in\mathbb{R}^d:\varphi(t)\leq c_\alpha(\Sigma_0)\}$, so that $g_x(\beta)=\Gamma_{\text{ar}}(\beta x,\Sigma_0,\varphi)$.  In view of from Lemma \ref{lemma:aditional_anderson_cond}, $g_x$ is strictly decreasing provided $W$ is closed, bounded and centrally symmetric. Because $\varphi$ is a norm, it is centrally symmetric, and thus, so is $W$.  Moreover, the hypograph $\{(\intvar,c) : \varphi(\intvar)\le c\}$ is closed and therefore upper semicontinuous by the continuity of $\varphi$. This, in turn, implies that $W$ is closed. Finally, we can show that $W$ is bounded similarly as was done for the set $B_0$ in the proof of Theorem \ref{thm:unbiased_locl_alt}. Since this establishes that $x\mapsto \Gamma(x,\Sigma_0,\varphi)$ is ray-decreasing, we find that $\text{pr}\left\{\Gamma_{\textnormal{ar}}(U, \Sigma_0,\varphi) = c\right\} = 0$ for every $c \in \mathbb{R}$ by Lemma \ref{lemma:zero_measure_ray_monotonic}.\\
\textbf{C3.} 
For any sequence $x_1,x_2,\ldots$ of elements in $\mathbb{R}^d$ with $\varphi(x_s)  \to \infty$, we have that
\begin{align*}
   \Gamma_{\textnormal{ar}}(x_s, \Sigma_0, \varphi)\ &=\ \pr\left\{\varphi(U_0+x_s) < c_\alpha(\Sigma_0) \right\}\\
   &\leq\ \pr\left\{\varphi(U_0)+\varphi(x_s) < c_\alpha(\Sigma_0)\right\}\ =\ 1 - \pr\left\{ c_\alpha(\Sigma_0) - \varphi(U_0) \leq \varphi(x_s)\right\}
\end{align*}by the triangle inequality. Because the random variable $c_\alpha(\Sigma_0) - \varphi(U_0)$ is bounded in probability, it follows from the above inequality that $\Gamma_{\text{ar}}(x_s,\Sigma_0,\varphi)$ tends to zero since  $\varphi(x_s) \rightarrow \infty$.

Now, suppose that there is no $\varepsilon>0$ over which, for every sequence $x_s$  for which $\varphi(x_s) \to \infty$, $\Gamma_{\textnormal{ar}}(x_s, \Sigma, \varphi) \to 0$ uniformly over all $\Sigma$ in a neighborhood of $\Sigma_0$.  There must then exist some $\delta>0$ and sequences $\Sigma_1,\Sigma_2,\ldots$ and $x_1,x_2,\ldots$ such that $\varphi(x_s) \to \infty$ and $\Sigma_s \to \Sigma$ but $\Gamma_{\text{ar}}(x_s,\Sigma_s,\varphi) > \delta$ for every $s$. By the continuity of $\varphi$ and $c_\alpha$, we have that $c_\alpha(\Sigma_s) - \varphi(U_s) \rightsquigarrow c_\alpha(\Sigma_0) - \varphi(U_0)$, where $U_0,U_1,U_2,\ldots$ is a sequence of independent random $d$-vectors with $U_s$ following the multivariate normal distribution with mean zero and covariance $\Sigma_s$. By Lemma 2.11 of \cite{van_der_vaart_asymptotic_2000}, this implies the uniform convergence of the corresponding distribution functions, and so, it follows that
\begin{align*}
  &|\pr\left\{ c_\alpha(\Sigma_0) - \varphi(U_0) \leq \varphi(x_s)\right\} - \pr\left\{ c_\alpha(\Sigma_s) - \varphi(U_s) \leq \varphi(x_s)\right\}|\\
  &\hspace{.5in}\leq\ \sup_{x}\left|\pr\left\{ c_\alpha(\Sigma_0) - \varphi(U_0) \leq x\right\} - \pr\left\{ c_\alpha(\Sigma_s) - \varphi(U_s) \leq x \right\}\right| \to 0
\end{align*}
Since we have already established above that $\pr\left\{ c_\alpha(\Sigma_0) - \varphi(U_0) \leq \varphi(x_s)\right\}\to 1$, it must then also be that $\pr\left\{ c_\alpha(\Sigma_s) - \varphi(U_s) \leq \varphi(x_s)\right\}\to 1$, and so, $\Gamma_{\textnormal{ar}}(x_s, \Sigma_s, \varphi) \leq 1 - \pr\left\{ c_\alpha(\Sigma_s) - \varphi(U_s) \leq \varphi(x_s)\right\}\rightarrow 0$.  This is a contradiction. As such, there must exist some neighborhood of $\Sigma_0$ such that the convergence of $\Gamma_{\text{ar}}(x_s,\Sigma,\varphi)$ to zero is uniform over $\Sigma$ in this neighborhood.
\\
\textbf{C4.} Let $x\in\mathbb{R}^d$ be given. Defining $A_{x} := \{\omega\in\mathbb{R}^d : \varphi(\omega + x) < c_\alpha(\Sigma_0)\}$, we note that  
\begin{align*}
	\gengam[\textnormal{ar}] & = \int I\left\{\varphi(\intvar) < c_\alpha(\Sigma_0)\right\}  f_{\Sigma_0}(\intvar - \garg)\,dt \\
	&= \int I\left\{\varphi(u+x) < c_\alpha(\Sigma_0)\right\}  f_{\Sigma_0}(u)\,du\ =\ \pr(U_0 \in A_{\garg})\ .
\end{align*}
Suppose that $x_1, x_2\in\mathbb{R}^d$ are such that $\Gamma_{\textnormal{ar}}(x_1, \Sigma_0,\varphi) \geq c$ and $\Gamma_{\textnormal{ar}}(x_2, \Sigma_0,\varphi) \geq c$. Then, we can write that $
	c = c^{t}c^{1 - t} \leq \Gamma_{\textnormal{ar}}(x_1, \Sigma_0,\varphi)^t \Gamma_{\textnormal{ar}}(x_2, \Sigma_0,\varphi)^{1 - t}$.
Theorem 1 of \citet{rinott_convexity_1976} states that $
     \nu(A_{ y})^t \nu(A_{ z})^{1 - t} \leq  \nu(t A_{ y} + (1 - t)A_{ z})$ for any distribution $\nu$ with log-concave density function,
where  $t A_{ y} + (1 - t)A_{ z}:= \{t  \omega_1 + (1 - t) \omega_2:  \omega_1 \in A_{ y},  \omega_2 \in A_{ z}, t \in [0, 1] \}$.
The multivariate normal distribution has a log-concave density, as shown, for example (see, e.g., Theorem 4.2.1 of \citealp{tong_multivariate_2012}), and so, it holds that $
     \pr(U_0\in A_{x_1})^t \pr(U_0\in A_{x_2})^{1-t}\leq \pr\{U_0\in tA_{x_1}+(1-t)A_{x_2}\}$. It remains to show that 
$\pr\{U_0\in tA_{x_1}+(1-t)A_{x_2}\}=\Gamma_{\text{ar}}(tx_1+(1-t)x_2,\Sigma_0,\varphi)=\pr\{U_0\in A_{tx_1+(1-t)x_2}\}$.
This is implied by Lemma \ref{lemma:set_convexity} and the fact that each $A_{ x}$ is convex by the convexity of norms, since this lemma shows that $tA_{\garg_1} + (1 - t) A_{\garg_2} = A_{t\garg_1 + (1 - t) \garg_2}$. Thus, we obtain that
\begin{align*}
	c\ \leq\ \pr\left\{U_0\in tA_{x_1}+(1-t)A_{x_2}\right\}\ =\ \pr\left\{U_0\in A_{tx_1+(1-t)x_2}\right\}\ =\ \Gamma_{\textnormal{ar}}(t\garg_1 + (1 - t) \garg_2, \Sigma_0,\varphi)\ .
\end{align*}
Thus, we have established that $x\mapsto \Gamma_{\textnormal{ar}}(x, \Sigma_0,\varphi)$ is quasi-concave.\\
\textbf{C5.} 
 In view of the facts that $U_0$ and $-U_0$ have the same distribution and that $\varphi$ is centrally symmetric, for any $x\in\mathbb{R}^d$, we have that
\begin{align*}
	\Gamma_{\text{ar}}(x,\Sigma_0,\varphi)\ =\ \pr\left\{\varphi(U_0 + x) < c_{\alpha}(\Sigma_0)\right\}\ &=\ 
	\pr\left\{\varphi(-U_0 + x) < c_{\alpha}(\Sigma_0)\right\}\\
	&=\ \pr\left\{\varphi(U_0 - x) < c_{\alpha}(\Sigma_0)\right\}\ =\  \Gamma_{\text{ar}}(-x, \Sigma_0,\varphi)\ .
\end{align*}

\noindent \underline{\textbf{Part 2:}} multiplicative factor measure.\\
\textbf{C1.} Again, non-negativity is clear. For $( x,\Sigma)\in \mathbb{R}^d \times \mathbb{R}^{d\times d}$, define $\Lambda_{ x,\Sigma} : \mathbb{R}^+ \rightarrow (0,1-\alpha]$ pointwise as $\Lambda_{ x,\Sigma}(s):=\Gamma_{\textnormal{ar}}(s x, \Sigma,\varphi)$. Since $x\mapsto \Gamma_{\textnormal{ar}}(x, \Sigma,\varphi)$ is continuous and strictly ray-decreasing, $s\mapsto \Lambda_{ x,\Sigma}(s)$ is also continuous and strictly decreasing. We note that $\Lambda_{ x,\Sigma}(0)=1-\alpha>\tau$ and $\lim_{s\rightarrow\infty} \Lambda_{ x,\Sigma}(s)=0$, and therefore,  $\Gamma_{\textnormal{mf}}(\garg,\Sigma,\varphi)=\min\{s\ge 0 :  \Lambda_{\garg, \Sigma}(s) \le \tau\}$ equals the inverse $\Lambda_{ x,\Sigma}^{-1}(\tau)$ of $\Lambda_{ x,\Sigma}$ at $\tau$.

Let sequences $x_1,x_2,\ldots\in\mathbb{R}^d$ and $\Sigma_1,\Sigma_2,\ldots\in\mathbb{V}^d$ such that $(x_k,\Sigma_k)\rightarrow (x,\Sigma)$ be given, and denote $\Lambda_k:=\Lambda_{ x_k,\Sigma_k}$ for each $k$ and $\Lambda:=\Lambda_{x,\Sigma}$. The continuity of $x\mapsto \Gamma_{\textnormal{ar}}(x,\Sigma,\varphi)$ implies that $\Lambda_k(s)\rightarrow \Lambda(s)$ for each $s$. In view of the continuity and monotonicity of the bounded functions $\Lambda_1,\Lambda_2,\ldots$ and $\Lambda$, an adaptation of arguments used to prove Lemma~2.11 of \cite{van_der_vaart_asymptotic_2000} can be used to show that $
    \sup_{s\ge 0}|\Lambda_k(s)- \Lambda(s)|\rightarrow 0$.
We prove by contradiction that $\Lambda_k^{-1}(\tau)\rightarrow \Lambda^{-1}(\tau)$. Suppose this is not so. Then, there exists  $\epsilon>0$ and natural numbers $k_1< k_2< \ldots$ such that (i) $\inf _j[\Lambda_{k_j}^{-1}(\tau)-\Lambda^{-1}(\tau)]\geq \epsilon$ or (ii)  $\inf_j[\Lambda^{-1}(\tau)-\Lambda_{k_j}^{-1}(\tau)]\geq \epsilon$.
Suppose that (i) holds. By the monotonicity of $\Lambda_k$, we have that $\tau<\Lambda_{k_j}(\Lambda^{-1}(\tau)+\epsilon)$ for all $j$, and so, 
\begin{align*}
    \tau\ &<\ \Lambda(\Lambda^{-1}(\tau)+\epsilon) + \Lambda_{k_j}(\Lambda^{-1}(\tau)+\epsilon) - \Lambda(\Lambda^{-1}(\tau)+\epsilon)\\
    &\le\ \Lambda(\Lambda^{-1}(\tau)+\epsilon) + \sup_{s\ge 0}|\Lambda_{k_j}(s)- \Lambda(s)|\ .
\end{align*}
As $\Lambda$ is strictly decreasing, $\Lambda(\Lambda^{-1}(\tau)+\epsilon)<\tau$. This yields a contradiction since the latter summand has been shown to tend to zero. A similar argument can be made if (ii) holds instead. We have thus shown that $x\mapsto \Gamma_{\textnormal{mf}}(x,\Sigma,\varphi)$ is continuous.\\
\textbf{C2.} Let $ v\in\mathbb{R}^d$ be given, and define $g_v : \mathbb{R} \to \mathbb{R}$ pointwise as \[g_{ v}(\beta) := \Gamma_{\textnormal{mf}}(\beta  v, \Sigma_0,\varphi)=\min\{s\geq 0 : \pr\left\{\varphi(U_0 + s\beta  v) \geq c_{\alpha}(\Sigma_0)\right\} \geq 1-\tau\}\ .\]
We note that $g_{ v}(\beta) = g_v(1) / \beta$, and so,  $x\mapsto\Gamma_{\textnormal{mf}}(x,\Sigma_0,\varphi)$ is strictly ray-decreasing. Hence, the conditions of Lemma \ref{lemma:zero_measure_ray_monotonic} are satisfied, and it follows that $\pr\{\Gamma_{\textnormal{mf}}(U_0, \Sigma_0,\varphi) = c\} = 0$ for each $c\geq 0$.\\
\textbf{C3.} 
Let a sequence $x_1,x_2, \dots\in\mathbb{R}^d$ such that $\varphi(\garg_x) \to \infty$ be given. Let $\varepsilon > 0$ be given, and set $\tilde{x}_j:=\varepsilon x_j$ for each $j$.  The sequence $\tilde{x}_1,\tilde{x}_2,\ldots$ also has the property that $\varphi(\tilde{x}_j)=\varphi(\varepsilon x_j) = \varepsilon \varphi(x_j) \to \infty$. Using condition \ref{itm:cond_conv_prb_zr} established in Part 1, there exists some $N>0$ and a neighborhood $B_0$ of $\Sigma_0$ such that $\Gamma_{\textnormal{ar}}(\tilde{x}_j, \Sigma, \varphi) \leq \tau$ for each $n > N$ and $\Sigma \in B_0$.  As $\Gamma_{\textnormal{ar}}(\tilde{x}_j, \Sigma, \varphi) \to 0$ and $\Gamma_{\text{mf}}(x_j, \Sigma, \varphi)$ is defined as the smallest $s$ such that $\Gamma_{\textnormal{ar}}(s x_j, \Sigma, \varphi) \leq \tau$, it follows that $\limsup_j \Gamma_{\text{mf}}(x_j, \Sigma, \varphi) \leq \varepsilon$ uniformly over $\Sigma\in B_0$.  Since $\varepsilon>0$ is arbitrary, it must be the case that $\Gamma_{\text{mf}}(x_j, \Sigma, \varphi) \to 0$ uniformly over $\Sigma\in B_0$.\\
\textbf{C4.} Suppose that $x_1, x_2\in\mathbb{R}^d$ are such that $\Gamma_{\textnormal{mf}}(x_1, \Sigma_0,\varphi) \geq c$ and $\Gamma_{\textnormal{mf}}(x_2, \Sigma_0,\varphi) \geq c$. When establishing condition  \ref{itm:cond_dom_by_lebesgue} in Part 1, it was shown that $s\mapsto \pr\{\varphi(U_0 + sx) \geq c_\alpha(\Sigma_0)\}$ is continuous and strictly increasing. Hence, if $\Gamma_{\textnormal{mf}}(x, \Sigma_0,\varphi) \geq x$, then $\pr\{\varphi(U_0 + cx) \geq c_\alpha(\Sigma_0)\} \geq 1-\tau$, which implies that $\pr\{\varphi(U_0 + cx_1) < c_\alpha(\Sigma_0)\} > \tau$ and $\pr\{\varphi(U_0 + cx_2) < c_\alpha(\Sigma_0)\} > \tau$. Using condition \ref{itm:cond_unimod} established in Part 1, we find that $\pr\{\varphi(U_0 + c \{tx_1 + (1 - t)x_2\}) < c_\alpha(\Sigma_0)\} > \tau$ or, equivalently, $\pr\{\varphi(U_0 + c \{tx_1 + (1 - t)x_2\}) \geq c_\alpha(\Sigma_0)\} \leq 1-\tau$. 
Thus, it follows that $\Gamma_{\textnormal{mf}}(tx_1 + (1 - t)x_2, \Sigma_0,\varphi) \geq c$, and so, $x\mapsto \Gamma_{\text{mf}}(x,\Sigma_0,\varphi)$ is quasi-concave. \\
\textbf{C5.} Using the fact that the density function of a mean-zero multivariate normal distribution is centrally symmetric, we have that
\begin{align*}
		\Gamma_{\textnormal{mf}}(x, \Sigma_0)\ &=\ \min\left\{s\geq 0 : \pr\{\varphi(U_0 + sx) \geq c_\alpha(\Sigma_0)\} \geq 1-\tau \right\}\\
		&=\ \min\left\{s\geq 0 : \pr\{\varphi(-U_0 + sx) \geq c_\alpha(\Sigma_0)\} \geq 1-\tau \right\}\\
		&=\ \min\left\{s\geq 0 : \pr\{\varphi(U_0 - sx) \geq c_\alpha(\Sigma_0)\} \geq 1-\tau \right\}\ =\  \Gamma_{\textnormal{mf}}(-x, \Sigma_0)
\end{align*}for each $x\in\mathbb{R}^d$, thereby establishing that $x\mapsto\Gamma_{\text{mf}}(x,\Sigma_0,\varphi)$ is centrally symmetric.
\end{proof}

\subsection*{Additional technical lemma}

The sum-of-squares function $\jmath_k : \mathbb{R}^d \to \mathbb{R}$ is defined as $x\mapsto \sqrt{\sum_{i = 1}^k x^2_{(d - i + 1)}}$, 
where $x^2_{(j)}$ is the $j^{th}$ order statistic of components of $x$.

\begin{lemma}
\label{lemma:ssq_norm}
The function $\jmath_k$ is a norm for each $k\in\{1,2,\ldots,d\}$.
\end{lemma}

\begin{proof}[Proof of Lemma~\ref{lemma:ssq_norm}]
Fix $k\in\{1,2,\ldots,d\}$. We must show that $\jmath_k$ is point-separating, absolutely homogeneous and subadditive, which then implies the claim. First, we note that if $\jmath_k(x) = 0$, then it must be that $0\leq x^2_{(1)}\leq\ldots\leq x^2_{(d)}\leq \sum_{i=1}^{k}x^2_{(d-i+1)}=0$, and so, $x_1=x_2=\ldots=x_d=0$. Second, we note that, for any $a > 0$ and $x\in\mathbb{R}^d$,
\begin{align*}
\jmath_k(ax) = \sqrt{\sum_{i = 1}^k \{ax_{(d - i + 1)}}\}^2  = \sqrt{\sum_{i = 1}^k a^2 x^2_{(d - i + 1)}} = a\sqrt{\sum_{i = 1}^k x^2_{(d - i + 1)}} = a\cdot\jmath_k(x) \ .
\end{align*} Finally, we let $ x = (x_1,x_2, \dots, x_d)^\top$ and $ y = (y_1,y_2, \dots, y_d)^\top $ be elements of $\mathbb{R}^d$, and define $z:=x+y$. Without loss of generality, suppose that $|z_1|\le |z_2|\le \ldots\le |z_d|$. Then, we have that
\begin{align*}
    \jmath_k( z)\ &=\ \sqrt{\sum_{i = 1}^k z^2_{(d - i + 1)}}\ =\ \sqrt{\sum_{i = 1}^k z^2_{d - i + 1}}\ =\ \sqrt{\sum_{i = 1}^k \left( x_{d - i + 1} + y_{d - i + 1}\right)^2} \\
    &\le\ \sqrt{\sum_{i = 1}^k x_{d - i + 1}^2} + \sqrt{\sum_{i = 1}^k y_{d - i + 1}^2}\ \le\ \sqrt{\sum_{i = 1}^k x_{(d - i + 1)}^2} + \sqrt{\sum_{i = 1}^k y_{(d - i + 1)}^2}\ =\ \jmath_k( x) + \jmath_k( y)\ ,
\end{align*}
where the first inequality follows from the subaddativity of the $\ell_2$ norm on $\mathbb{R}^k$.
\end{proof}

\subsection*{Additional figures}

\begin{figure}
	\centering
\includegraphics[width = \linewidth]{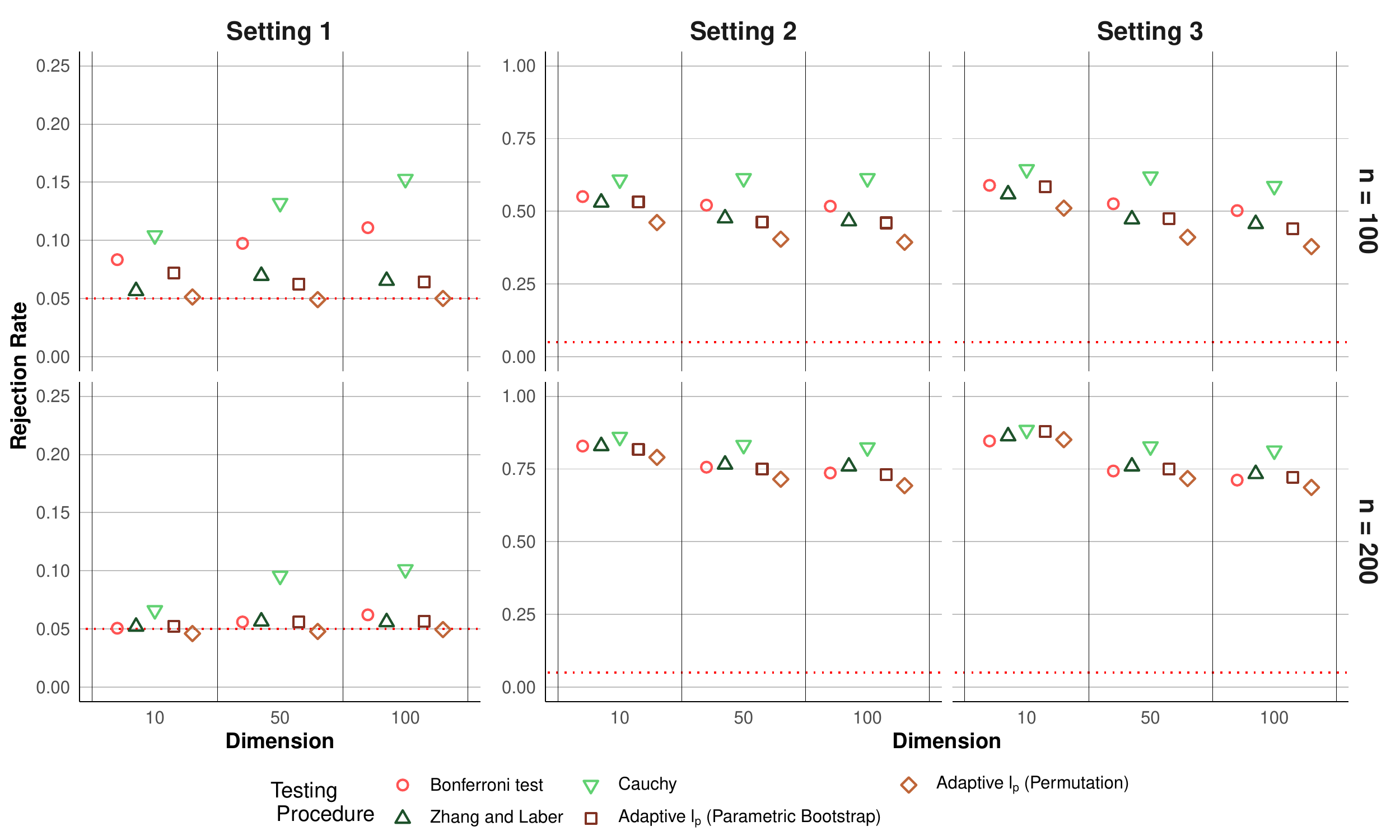}
	\caption{Simulation study-based empirical rejection rate of various tests applicable in Example 1 under different data-
generating mechanisms, at different sample sizes, and for covariate vectors with moderate correlation (50\%)
across components and of different length.
}
	\label{fig:some_cor}
\end{figure}

\begin{figure}
	\centering
\includegraphics[width = \linewidth]{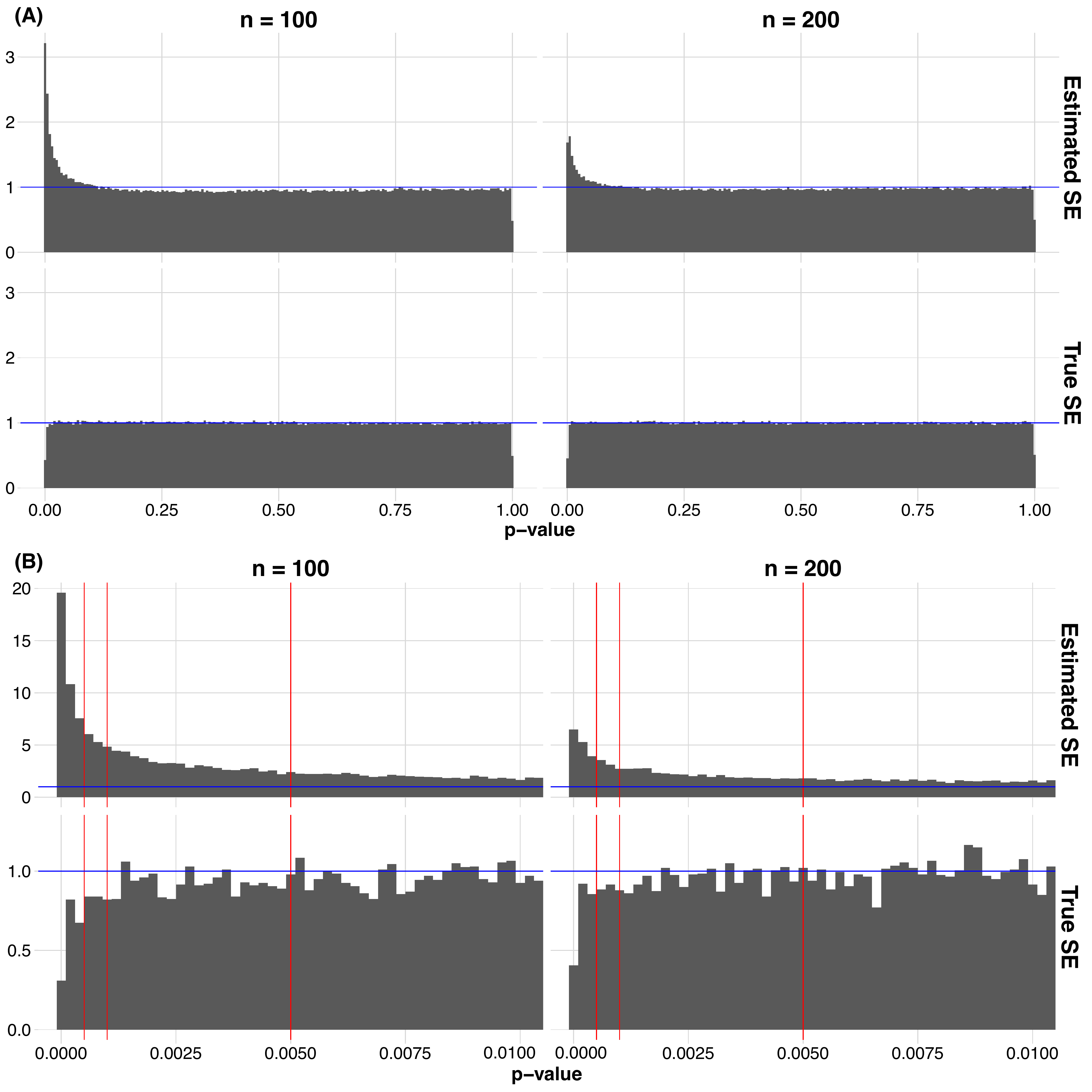}
	\caption{Simulation-based empirical sampling density  of the $p$-value $p_{1n}(\sigma):=2\,[1-\Phi(n^{1/2}|\psi_{n1}|/\sigma)]$ for $\sigma$ equal to either the true asymptotic standard deviation $\sigma_0$ or its influence function-based estimator $\sigma_n$ based on Example 1.  Data are generated from the setting in which all covariates are independent of each other and the outcome.  Panel (A) shows sampling densities on $[0, 1]$. Panel (B) shows the same densities but is restricted to the interval $[0, 0.01]$. In each panel, displays in the top and bottom rows show, respectively, the sampling density when $\sigma_0$ is estimated or known. Displays in the left and right columns show, respectively, results for $n = 100$ or $n = 200$. The blue horizontal line represents the theoretical standard uniform density of $p$-values under the null, and the red vertical lines (left to right) in Panel (B) are the largest single covariate $p$-value that results in rejection of the Bonferroni test for dimension $d$ equal to 100, 50 and 10.}
	\label{fig:pval_all}
\end{figure}

\begin{figure}
	\centering
\includegraphics[width = \linewidth]{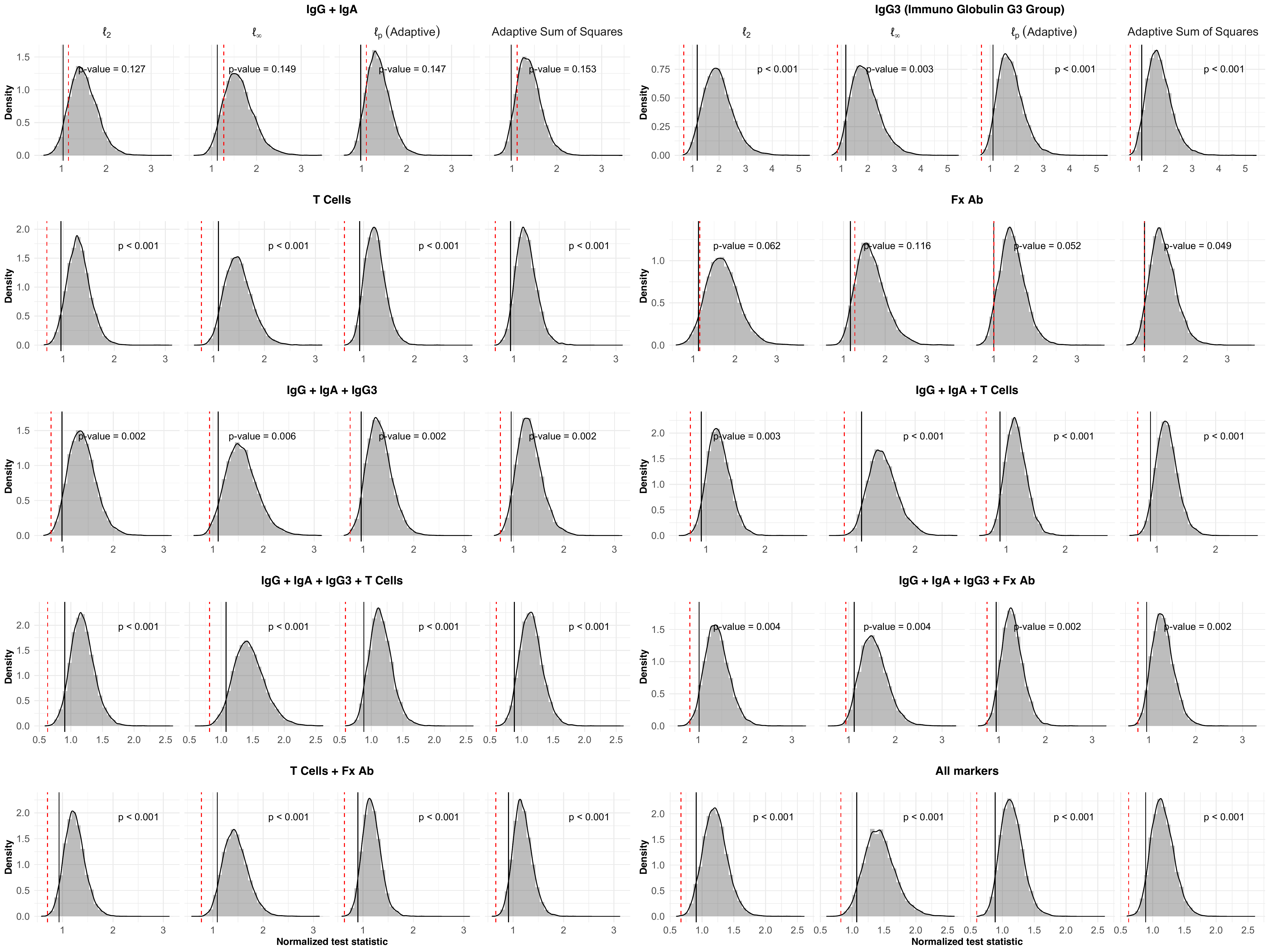}
	\caption{Estimated limiting distribution of the multiplicative factor measure for both non-adaptive ($\ell_2$ and maximum absolute deviation) and adaptive (adaptive $\ell_p$ and adaptive sum-of-squares) testing procedures.  The black and dashed red vertical lines in each plot denote the value of the $5^{th}$ percentile of the limiting distribution and of the test statistic, respectively. This analysis is based on data from the HVTN 505 clinical trial, and the null hypothesis tested is that the biomarkers from the Fx Ab group are not associated with risk of HIV infection.}
	\label{fig:all_de_figs}
\end{figure}

\subsection*{Additional information on data analysis}
\label{sec:supp_data_app}

During the HVTN 505 clinical trial, vaccinations were administered at Months 0, 1 and 6.  To measure the immune response to vaccination, samples were collected from individuals at Month 7. After trial completion, samples were analyzed for 25 primary endpoint vaccine cases (HIV-1 infected between Month 7 and Month 24) and 125 randomly sampled frequency matched vaccine controls (HIV-1 negative at Month 24) \citep{janes_higher_2017}. Baseline covariates and infection status were recorded for all study participants.

We denote the vector of recorded immune response biomarkers as $S := (S_1,S_2, \ldots, S_d)$, and denote by $W$ and $Y$ the baseline covariate vector and infection status, respectively. The biomarker vector $S$ is only recorded on a subset of participants, and the variable $\Delta$ indicates those patients, with $\Delta=1$ if $S$ is recorded and 0 otherwise. For each group of biomarkers considered in \cite{neidich_antibody_2019}, we test the null hypothesis that these biomarkers are not associated with risk of infection. The measure of association used for each biomarker $S_j$ is the $\beta_1$-coefficient value indexing the KL projection of the conditional log-odds of infection $ \log\left(\text{odds}\left(Y = 1 | S_j = s\right)\right)$ onto a linear working model $\beta_1+\beta_2 s$.

Denoting by $P_F$ a candidate distribution for the full-data unit $(W,S,Y)$, we first define the full-data parameter for biomarker $S_j$ to be
\begin{align*}
  \boldsymbol{\beta}_{j,F}(P_F):=\argmax_{\boldsymbol{\beta}}\, 
  E_{P_F}\left[R_j(\boldsymbol{\beta})(S_j,Y)\right]\ ,
\end{align*}where we also define $R_j(\boldsymbol{\beta}):(s,y)\mapsto y\log\{\text{expit}(\beta_1 + \beta_2 s)\}+(1-y)\log\{1 - \text{expit}(\beta_1 + \beta_2 s)\}$.
For the two-phase design, the observed-data unit is $X:=(W,\tilde{S},\Delta,Y)$ with $\tilde{S}:=\Delta S$. Each participant's probability of being sampled in the second phase depends on their outcome and baseline covariate vector but not on the biomarker vector. In other words, $\Delta$ and $S$ are independent given $(W,Y)$.  In this particular study, all cases were sampled but controls were sampled based on BMI, race and ethnicity \citep{janes_higher_2017}. Under this assumption, the full-data parameter can be expressed as the observed-data parameter
\begin{align*}
\boldsymbol{\beta}_j(P):=\argmax_{\boldsymbol{\beta}} E_P\left[\frac{\Delta }{P(\Delta=1\mid Y,W)}\cdot R_j(\boldsymbol{\beta})(\tilde{S}_j,Y)\right]\ ,
\end{align*}where $P$ is a candidate distribution of the observed-data unit.

In the context considered, $(W,Y)$ have a finite support under the true sampling distribution $P_0$. Thus, the parameter value $\boldsymbol{\beta}_{j0}:=\boldsymbol{\beta}_j(P_0)$ can be estimated using the plug-in estimator $\boldsymbol{\beta}_{jn}:=\boldsymbol{\beta}_j(P_n)$, where $P_n$ is the empirical distribution based on $X_1,X_2,\ldots,X_n$; in practice, this estimator can be obtained using weighted univariable logistic regression with empirically computed weights. The estimator $\boldsymbol{\beta}_{jn}$ is a vector $(\beta_{jn,1},\beta_{jn,2})$, with components giving estimators of the constant and slope of the best linear model approximation to the true conditional log-odds of risk of HIV infection, respectively.

The influence function of $\boldsymbol{\beta}_{jn}$ is given by \[x\mapsto -M_{j0}^{-1}\left[\frac{\delta}{\pi_0(w,y)}\nabla R_j(\boldsymbol{\beta}_{j0})(w,\tilde{s},y)+\left\{1-\frac{\delta}{\pi_0(w,y)}\right\}\xi_{j0}(w,y)\right]\,,\]where we define pointwise the nuisance  functions $\pi_0(w,y):=E_0\left(\Delta\mid W=w,Y=y\right)$ and $\xi_{j0}(w,y):= E_0\left[\nabla R_j(\boldsymbol{\beta}_{j0})(W,\tilde{S},Y)\mid \Delta=1,W=w,Y=y\right]$ as well as the normalization matrix \[M_{j0}:=E_0\left[\frac{\Delta}{\pi_0(W,Y)}\nabla^2R_j(\boldsymbol{\beta_{j0}})(W,\tilde{S},Y)\right].\]Here, defining $m_{\boldsymbol{\beta}}:s\mapsto \expit(\beta_1+\beta_2s)$, we can compute $\nabla R_j(\boldsymbol{\beta})(w,s,y)=[y-m_{\boldsymbol{\beta}}(s)]\left[\begin{smallmatrix}1\\ s\end{smallmatrix}\right]$ and $\nabla^2R_j(\boldsymbol{\beta})(w,s,y)=-m_{\boldsymbol{\beta}}(s)[1-m_{\boldsymbol{\beta}}(s)]\left[\begin{smallmatrix}1 & s\\ s & s^2
\end{smallmatrix}\right]$. In particular, the influence function of $\psi_{jn}:=\beta_{j2,n}$ is given by 
\begin{align*}
    \phi_{j0}:x\mapsto\ & \frac{\delta}{\pi_0(w,y)}\left\{a_{j0}+b_{j0}\tilde{s}\right\}\left\{y-m_{\boldsymbol{\beta}_{j0}}(\tilde{s})\right\}\\
    &+\left\{1-\frac{\delta}{\pi_0(w,y)}\right\}
    E_0\left[\big{(}a_{j0}+b_{j0}\tilde{S}\big{)}\big{\{}Y-m_{\boldsymbol{\beta}_{j0}}(\tilde{S})\big{\}}\,\middle|\, \Delta=1,W=w,Y=y\right]\ ,
\end{align*}where $a_{j0}$ and $b_{j0}$ are the $[2,1]$ and $[2,2]$ entries of $-M_{j0}^{-1}$. This implies that $n^{1/2}\left(\psi_n-\psi_0\right)$ converges in distribution to a mean-zero multivariate normal distribution with covariance matrix $\Sigma_0$ with entry $[j,k]$ given by $\Sigma_{jk}:=E_0\left[\phi_{j0}(X)\phi_{k0}(X)\right]$. As such, a natural estimator $\Sigma_n$ of $\Sigma_0$ is defined entrywise as $\Sigma_{jk,n}:=\frac{1}{n}\sum_{i=1}^{n}\phi_{jn}(X_i)\phi_{kn}(X_i)$ with 
\begin{align*}
    \phi_{jn}:x\mapsto\ & \frac{\delta}{\pi_n(w,y)}\left\{a_{jn}+b_{jn}\tilde{s}\right\}\left\{y-m_{\boldsymbol{\beta}_{jn}}(\tilde{s})\right\}\\
    &+\left\{1-\frac{\delta}{\pi_n(w,y)}\right\}
    E_n\left[\big{(}a_{jn}+b _{jn}\tilde{S}\big{)}\big{\{}Y-m_{\boldsymbol{\beta}_{jn}}(\tilde{S})\big{\}}\,\middle|\, \Delta=1,W=w,Y=y\right]\ ,
\end{align*}where $\pi_n$ is an estimator of $\pi_0$, $a_{jn}$ and $b_{jn}$ are the $[2,1]$ and $[2,2]$ entries of $-M_{jn}^{-1}$ with $M_{jn}:=-\frac{1}{n}\sum_{i=1}^{n}\frac{\Delta_i}{\pi_n(W_i,Y_i)}m_{\boldsymbol{\beta}_{jn}}(\tilde{S}_i)[1-m_{\boldsymbol{\beta}_{jn}}(\tilde{S}_i)]\left[\begin{smallmatrix}1 & \tilde{S}_i\\ \tilde{S}_i & \tilde{S}_i^2
\end{smallmatrix}\right]$, and $E_n$ denotes an empirical expectation relative to the distribution of $\tilde{S}$ given $\Delta=1$ and $(W,Y)$.

\end{document}